\documentclass{article}
\usepackage[margin=1in]{geometry}

\usepackage{graphicx} % Required for inserting images
\usepackage{amsfonts}
\usepackage{amsthm}
\usepackage{thm-restate}
\usepackage{amssymb}
\usepackage{amsmath}
\usepackage{dsfont}
\usepackage{tikz}
\usepackage{minted}
\definecolor{citec}{HTML}{324FDA} 
\definecolor{linkc}{HTML}{A0111A}
\definecolor{urlc}{HTML}{b55c87}
\usepackage[
    colorlinks,
    citecolor=citec,
    linkcolor=linkc,
    urlcolor=urlc,
    bookmarks = true,
    breaklinks,
]{hyperref}
\usepackage{cleveref}
\usepackage{cite}
\usepackage{wrapfig}
\usepackage{comment}
\usepackage[super]{nth}

\usepackage{mathtools}

\usepackage{enumitem}
\newlist{todolist}{itemize}{2}
\setlist[todolist]{label=$\square$}

\usepackage{pifont}

\usetikzlibrary{arrows.meta, positioning}

\newtheorem{claim}{Claim}
\newtheorem{remark}{Remark}
\newtheorem{lemma}{Lemma}
\newtheorem{theorem}{Theorem}
\newtheorem{proposition}{Proposition}
\newtheorem{definition}{Definition}

\newtheorem{corollary}{Corollary}
\newtheorem{example}{Example}

\newcommand{\Var}{\mathrm{Var}}
\newcommand{\Cov}{\mathrm{Cov}}

\newcommand{\calX}{\mathcal{X}}
\newcommand{\calY}{\mathcal{Y}}

\newcommand{\calB}{\mathcal{B}}
\newcommand{\calP}{\mathcal{P}}

\newcommand{\bfv}{{\boldsymbol{v}}}
\newcommand{\bfe}{{\boldsymbol{e}}}

\newcommand{\CCP}{\text{CCP}}
\newcommand{\Gumbel}{\mathrm{Gumbel}}

\newcommand*\samethanks[1][\value{footnote}]{\footnotemark[#1]}

\newcommand\blfootnote[1]{%
  \begingroup
  \renewcommand\thefootnote{}\footnote{#1}%
  \addtocounter{footnote}{-1}%
  \endgroup
}

\title{Expected Recovery Time in DNA-based Distributed Storage Systems}

\author{Adi Levy\thanks{Department of Computer Science, %\\
    Technion---Israel Institute of Technology, Haifa 3200003, Israel, emails: \{levyadi, roni.con, yaakobi\}@cs.technion.ac.il} \and Roni Con\samethanks \and Eitan Yaakobi\samethanks \and Han Mao Kiah\thanks{School of Physical and Mathematical Sciences, Nanyang Technological University, Singapore, email: hmkiah@ntu.edu.sg}
    }
\date{}

\begin{document}

\maketitle
\begin{abstract}
    We initiate the study of \emph{DNA-based distributed storage systems}, where information is encoded across multiple DNA data storage containers to achieve robustness against container failures.
    In this setting, data are distributed over $M$ containers, and the objective is to guarantee that the contents of any failed container can be reliably reconstructed from the surviving ones.
    Unlike classical distributed storage systems, DNA data storage containers are fundamentally constrained by sequencing technology, since each read operation yields the content of  a uniformly random sampled strand from the container.
    Within this framework, we consider several erasure-correcting codes and analyze the expected recovery time of the data stored in a failed container.
    Our results are obtained by analyzing generalized versions of the classical Coupon Collector's Problem, which may be of independent interest.\blfootnote{A.L., R.C., and E.Y.  are supported by the European Union (DiDAX, 101115134). Views and opinions expressed are those of the
    authors only and do not necessarily reflect those of the European Union or the European Research Council
    Executive Agency. Neither the European Union nor the granting authority can be held responsible for them.}
\end{abstract}

\tableofcontents
\newpage

\section{Introduction}
Recent advances in DNA-based data storage technologies have positioned DNA as a compelling substrate for long-term digital archiving, offering density, durability, and longevity far beyond the capabilities of conventional media.
As global data generation accelerates at an unprecedented pace, existing storage technologies such as magnetic disks, solid-state drives, and tapes struggle to meet the demands of scalability, reliability, and archival lifetime. 
DNA, by contrast, provides an ultra-dense and chemically stable medium whose information can persist for thousands of years under modest conditions.
%and whose synthesis and sequencing technologies continue to improve rapidly. 
These features position DNA-based storage as a promising solution for the growing archival storage crisis \cite{church2012next,goldman2013towards,yazdi2016dna,organick2018random}.

A conventional DNA-based data storage system is composed of three main elements: DNA synthesis, storage containers, and DNA sequencing. First, artificial DNA strands, called oligos, are produced to encode the user’s data.
These strands are then stored in an unordered manner within a storage container, with each strand present in millions of copies.
Finally, DNA sequencing is used to read the stored strands and transform them into digital sequences, known as reads, which are then decoded to reconstruct the original user information. In each of these three components, errors might occur, and thus, much attention was given to designing error-correcting codes for DNA-based storage systems. We refer the reader to the following survey \cite{sabary2024survey} that presents the coding challenges in such systems.

In this work, we initiate the study of \emph{coding for DNA-based distributed storage systems}. We consider a setting with $M$ DNA storage containers, and our objective is to encode data across these $M$ containers so that, even if one or more containers fail, the lost information can still be reconstructed from the remaining containers.
The formulation of this problem matches that of standard distributed storage systems.  
However, there are a few important differences between our notion of repair and the one adopted in the classical storage setting.

In a classical distributed storage system, a file is divided into $k$ data fragments and then encoded using an erasure-correcting code, producing $M>k$ fragments in total, with each fragment stored on a different one of the $M$ nodes.
Since an erasure code is used, repairing a failed node reduces to decoding a single erasure of a codeword.
A central performance metric in traditional distributed storage systems is the \emph{repair bandwidth}, which denotes the amount of information downloaded from the remaining nodes to recover the failed one. 
The problem of minimizing the repair bandwidth was introduced in the seminal work of Dimakis et al. \cite{dimakis2010network} and has since attracted extensive attention, leading to ingenious constructions that achieve optimal tradeoffs (for a detailed overview, the reader is referred to the survey \cite{ramkumar2022codes}). 
We emphasize that in this model, nodes are allowed to perform arbitrary computations on their stored data and transmit the resulting outputs.

In DNA-based distributed storage, the process of downloading information from containers is fundamentally different. First, no computation can be performed within a container. Second, reading is carried out by a sequencer, which we assume samples a strand uniformly at random from the container in each read. 
Because each of the $n$ strands is present in the tube in about a million identical copies, all in equal numbers, drawing a strand uniformly at random from the container is effectively the same as choosing one of the $n$ distinct strands uniformly at random. 
%{\color{red} We stress that the reading process is carried out with replacement, meaning that strands remain in the container after being read, and hence each read samples a strand uniformly at random. 
%This model is referred to as \emph{the random access model}}.

It was observed in \cite{bar2024cover}, that this sampling model is deeply connected with the classical coupon collector's problem (CCP, for short). In the CCP, there are $n$ distinct coupons, and at each point of time, the collector picks a coupon uniformly at random from the $n$ coupons, and then returns it. 
The task of the collector is to see all the $n$ distinct coupons, and the question is how many rounds it takes.
This problem has been extensively investigated in combinatorics and probability theory; see, for instance, \cite{erdHos1961classical,boneh1997coupon,adler2001coupon,kobza2007survey,alistarh2021collecting}. However, the classical problem received many new twists with the new connection to DNA-based storage systems \cite{bar2024cover,abraham2024covering,gruica2024combinatorial,gruica2024geometry,boruchovsky2025making}. 

In this paper, we introduce a DNA-based distributed storage system model. We establish a connection between this model and a generalized coupon collector's problem involving multiple independent collectors, each with its own set of coupons to be collected. To the best of our knowledge, this generalized setting has not been studied previously and may be of independent interest.

\subsection{Distributed DNA Storage and Coupon Collector's Problem}
For a positive integer $n$, let $[n] \triangleq \{1,\ldots,n\}$.
We consider a distributed storage system with $M$ containers, each storing $n$ distinct strands.
Each container is equipped with its own sequencer, and reads are performed in parallel:
at each time step, one strand is sampled uniformly at random from each functioning container.
When one container fails, our goal is to reconstruct its data from the remaining $M-1$ containers,
and we study the expected time required for this recovery.
% The central question we study is: \emph{How long does it take to recover this data?}

There could be several reasons for a container to become unavailable. 
These include physical degradation or loss of DNA molecules due to chemical decay or handling errors~\cite{church2012next,goldman2013towards,yazdi2016dna}, as well as failures in the sequencing or library preparation pipeline that render the contents of a container unreadable despite the DNA remaining intact~\cite{organick2018random,sabary2024survey}. At scale, additional failure modes include contamination or mislabeling events that prevent reliable identification of a container’s data~\cite{yazdi2016dna,organick2018random}. From a coding-theoretic perspective, such events are naturally modeled as container-level erasures, in direct analogy to node failures in classical distributed storage systems~\cite{dimakis2010network,ramkumar2022codes}.
\begin{remark}
    We emphasize that, in practice, each strand is stored in millions of identical copies.
    For simplicity, and as an approximation of this setting, we assume that at any time, each container contains exactly $n$ distinct strands, from which one strand is sampled uniformly at random, read, and returned. 
\end{remark}
\begin{figure}[h]
    \centering
    \includegraphics[width=0.6\textwidth]{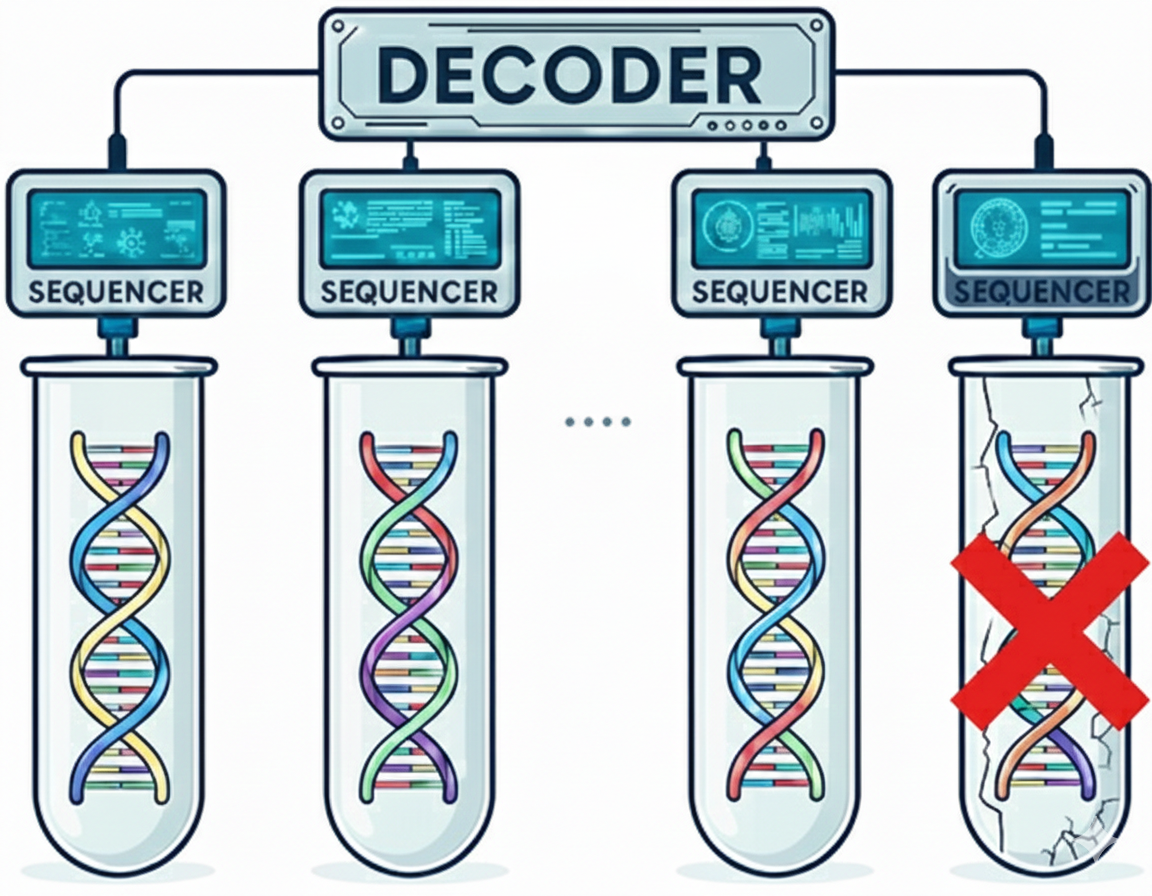} 
    \caption{An illustration of a DNA-DSS} 
    \label{fig:DNA-DSS}
\end{figure}
We begin by formally defining a DNA-based distributed storage system, and then introduce the main parameter of interest, namely the expected recovery time of a container.

\begin{definition}[$(n,M,k,|\Sigma|)$ DNA-based distributed storage system (DNA-DSS)] \label{def:DNA-DSS}
    %Suppose there are $M$ storage containers, each capable of storing $n$ distinct strands. Let $\Sigma$ be an alphabet and $k$ an integer.  
    Let $n, M$, and $k$ be positive integers and let $\Sigma$ be an alphabet.
    An $(n, M, k, |\Sigma|)$ \emph{DNA-based distributed storage system (DNA-DSS)} is defined by an injective encoder $\mathcal{E}: \Sigma^k \to \Sigma^{n\times M}$ such that, for every $\bfv \in \Sigma^k$ and every $j \in [M]$, the following condition is satisfied: The $j$-th column of $S \triangleq \mathcal{E}(\bfv)$ can be reconstructed from the collection of columns of $S$ indexed by $[M]\setminus\{j\}$.
\end{definition}
Note that this definition is purely mathematical and does not yet refer to DNA-based storage systems. 
The relation to DNA-based storage is as follows. Suppose $\bfv \in \Sigma^k$ represents the data we wish to store and set $S = \mathcal{E}(v)$. 
For each $j \in [M]$, the $j$-th container stores the $n$ entries of the $j$-th column of $S$. Thus, every container holds $n$ values, corresponding to the $n$ strands within that container.
The requirement in \Cref{def:DNA-DSS} can be rephrased as follows: the data in the failed $j$-th container can be reconstructed by reading all the data from the remaining $M-1$ containers.
\begin{definition}[The expected recovery time of a container]
    For every $j \in [M]$, we define  $T_j$ as the random variable representing the time required to recover the $j$-th container under this model. The \emph{expected recovery time of the $j$-th container} will be denoted by $\Bbb{E}[T_j]$.
\end{definition}

In this paper, we design DNA-DSS using \emph{MDS (maximum distance separable) codes} and analyze the expected recovery time of a failed container.
An MDS code of length $M$ and redundancy $r$ is an injective encoder
$\mathcal{E} : \Sigma^{M-r} \to \Sigma^M$
such that any $\mathbf v\in \Sigma^{M-r}$ can be uniquely recovered from any $M-r$ entries of the \emph{codeword} $\mathcal{E}(\mathbf v)$.
Before presenting our main results, we briefly revisit the classical coupon collector problem.
We shall see that several of our results can be viewed as generalizations of this problem.

{\bf{The Coupon Collector's Problem.}}
We give a formal definition of the classical coupon collector distribution, as well as of the distribution corresponding to the version with $\ell$ copies. 

\begin{definition}[Coupon Collector's distribution]
    \label{def:ccp-dist}
        Consider $n$ distinct coupon types, and draw one coupon at each trial independently and uniformly from these $n$ types. Let $X^{(n)}$ be the number of draws needed to obtain at least one of each type, and let $X^{(n,\ell)}$ be the number of draws needed to obtain at least $\ell$ of each type. We write $X^{(n)} \sim \text{CCP}(n)$ and $X^{(n,\ell)} \sim \text{CCP}(n,\ell)$.
    \end{definition}
It is a well-known fact that $\Bbb{E}[X^{(n)}] = n (\ln n + \gamma) + O(1)$, 
where $\gamma=0.577\ldots$ is the  Euler–Mascheroni constant.
(e.g., \cite[Section 2.4.1]{bookProbabilityAndComputing}). In \cite[Theorem 2]{newman1960double}, it was shown that when $\ell$ is fixed, we have $\mathbb{E}\left[X^{(n,\ell)}\right] = n\ln n + (\ell-1)n\ln \ln n +n\cdot C_\ell + o(n)$,
    and later, in \cite{erdHos1961classical}, it was shown that $C_{\ell} = \gamma - \ln ((\ell - 1)!)$.
\section{Our Results}
MDS codes have the best known rate-erasure correcting tradeoff and are therefore widely studied in classical distributed storage systems.
We likewise study the performance of MDS codes in the DNA-DSS setting.
Our results follow the classical split: first, we consider \emph{scalar MDS codes}, and then \emph{array regenerating MDS codes}.
We note, however, that MDS and regenerating codes may not be optimal for minimizing the expected container recovery time, and we leave this question for future work.

Throughout this paper, we work under the assumption that the number of containers, denoted by $M$, is fixed, while the number of distinct strands $n$ tends to infinity. In addition, whenever we use the notation $o(1)$, it represents a term that goes to $0$ as $n$ goes to infinity. Finally, for numbers $a, b$ and $c$, we will denote by $a \pm b$ the interval $[a-b, a+b]$ and $c(a\pm b) = ca \pm cb$.
\subsection{The \emph{Scalar} MDS Code Case}
We design an $(n, M, n(M-r), q)$ DNA-DSS based on an MDS code of length $M$ and redundancy $r$ as follows. 
Let $S' \in \Sigma^{n \times (M-r)}$ denote the information to be stored (in the notation of \Cref{def:DNA-DSS}, this corresponds to $|\Sigma| = q$ and $k = n(M-r)$).
Each row vector of $S'$ is encoded using the encoder of the MDS code into a row vector $\bfv$ of length $M$. 
Let $S \in \Sigma^{n \times M}$ denote the matrix whose rows are these encoded vectors, and store each column of $S$ in a distinct container.
Once a container fails, the MDS property guarantees that to recover its content, it is enough to know any $M - r$ symbols from each row of $S$. 

We therefore inquire for the expected time it takes to read this data from the functional containers. We will formulate this question, and obtain its answer, in terms of the following random process.

\begin{restatable}{theorem}{MDSVersion} \label{thm:mds-version}
    Let $A^{(0)} = {\boldsymbol{0}}^{n\times (m+\rho)}, A^{(1)}, A^{(2)}, \ldots$ be a sequence of matrices constructed as follows: for each $t \in \mathbb{N}$, we draw $(v_1, \ldots, v_{m+\rho}) \sim \mathrm{Unif}([n]^{m+\rho})$ and set
    \[
        A^{(t)} = A^{(t-1)} + \left[\bfe_{v_1}, \ldots, \bfe_{v_{m+\rho}}\right] \;,
    \]
    where $\bfe_a$ denotes the column vector with a $1$ in the $a$-th position and $0$ in all other entries. Define $T_{n,m,\rho}$ as the random variable that corresponds to the smallest $t$ such that every row of $A^{(t)}$ contains at least $m$ nonzero entries.
    Then,
    \[
    \mathbb{E}\left[T_{n,m,\rho}\right] = \frac{n}{\rho+1} \left( \ln n + \ln \binom{m+\rho}{m-1} + \gamma \pm o(1) \right).
    \]
    % \hanmao{This should be $\mathbb{E}\left[T_{n,m,\rho}\right]$ right? I didnt change because I wasnt sure how it looks like in the long version.}
\end{restatable}

The next corollary states the associated DNA-DSS, derived from \Cref{thm:mds-version} by instantiating it with a standard MDS code having redundancy $r = \rho + 1$ and length $M = m + \rho + 1$. As concrete choices, one may use a \emph{doubly extended Reed–Solomon code} over an alphabet of prime-power size \cite[Section 5, Problem 5.2]{roth2006introduction}, or, in the special case where a single container acts as the redundancy container, a simple parity code suffices.

% \hanmao{Small suggestion: In the previous statement, we write things in terms of $m$ and $\rho$. However, the next corollary uses $n$, $M$ and $r$. Maybe we can relate them? Like setting $r=\rho+1$ and $M=m+\rho+1$  in Theorem~\ref{thm:mds-version}?}
\begin{corollary} \label{cor:MDS-DDSS}
    Let $n, M, r \in \mathbb{N}$. If $r = 1$, choose any positive integer $q\geq 2$. If $r \geq 2$, choose $q$ to be a prime power satisfying $q \geq M - 1$.
    Then there exists an $(n, M, n(M - r), q)$ DNA-DSS such that the expected recovery time, $T_j$, of every container $j\in [M]$ is 
    \[
    \Bbb{E}[T_j] = 
    \frac{n}{r} \ln n + \frac{n}{r} \ln \binom{M-1}{r} + \frac{\gamma n}{r} \pm o(n) \;.
    \]
\end{corollary}

%We observe that for essentially all MDS codes of length $M$, the alphabet size $q$ must satisfy $q \geq M - 1$, with only a few exceptional cases (including the case $r=1$). The MDS conjecture \cite{segre1955curve}, which is believed to be true, asserts that aside from these known special situations, no additional MDS codes exist with $q < M - 1$.

{\bf{A direct theorem for $r=1$.}}
We observe that \Cref{thm:mds-version} and \Cref{cor:MDS-DDSS} already handle the situation in which exactly one container stores redundant data.
Nonetheless, to build intuition for the proof of \Cref{thm:mds-version}, we present here and prove in \Cref{sec:proof-of-one-redund}a simpler statement,
which can be viewed as a straightforward extension of the classical coupon collector problem.
In the special case $r = 1$, whenever a container fails, recovering it necessitates reading \emph{all} strands from each of the remaining containers.
This is equivalent to asking for the expected time until $m$ independent coupon collectors (each with its own set of $n$ coupons) collect all their coupons.
Formally, we show the following.
\begin{restatable}{theorem}{singleCopyNoRedund} \label{thm:single-copy-no-redun}
    Let $n\in \Bbb{N}$ and let $m$ be a constant positive integer with respect to $n$.
    Let $X^{(n)}_1, \ldots, X^{(n)}_{m}$ be i.i.d. random variables with $X^{(n)}_j \sim \text{CCP}(n)$. Then,
    \[
    \Bbb{E} \left[ \max_{j\in[m]} X^{(n)}_{j} \right] = n \left( \ln (mn) + \gamma  \pm o(1) \right) \;.
    \]
\end{restatable}

% XXXXXX \\\\\\\\\\

% We formulate this problem in equivalent form using the CCP terminology. 
% Suppose we have $m$ i.i.d. random variables $X^{(n)}_1, \ldots, X^{(n)}_{m}$ with $X^{(n)}_j \sim \text{CCP}(n)$ for every $j\in[m]$. What is the value of $\mathbb{E} \left[ \max_{j\in[m]} X^{(n)}_{j} \right]$?
% We obtain the following theorem.

% \begin{restatable}{theorem}{singleCopyNoRedund} \label{thm:single-copy-no-redun}
%     Let $n\in \Bbb{N}$ and let $m$ be a constant positive integer with respect to $n$.
%     Let $X^{(n)}_1, \ldots, X^{(n)}_{m}$ be i.i.d. random variables with $X^{(n)}_j \sim \text{CCP}(n)$. Then,
%     \[
%     \Bbb{E} \left[ \max_{j\in[m]} X^{(n)}_{j} \right] = n \ln (mn) + \gamma n \pm o(n)\;.
%     \]
% \end{restatable}

% A straitforward corollary is the following
% \begin{corollary} \label{cor:Parity-DNA-DSS}
%     For any $n, q \in \mathbb{N}$ and any $M$ that is constant with respect to $n$, there exists an $(n, M, n(M-1), q)$ DNA-DSS such that the expected recovery time of every failed container is $n \ln((M-1)n) + \gamma n \pm o(n)$.
% \end{corollary}

In practice, each strand is stored in many identical copies, so sequencing typically yields multiple reads of the same strand.  
When errors are present, having several corrupted versions of each strand can enhance the tradeoff between rate and reliability.  
In this work, however, we do not address errors.  
Instead, we generalize \Cref{thm:single-copy-no-redun} to the setting where $\ell$ copies of each coupon are required, laying the groundwork for future extensions that incorporate also storage errors.

\begin{restatable}{theorem}{MultipleCopyNoRedund} \label{thm:multi-copy-no-redun}
    Let $m$ and $\ell$ be constant positive integers.
    Let $X^{(n,\ell)}_1, \ldots, X^{(n,\ell)}_{m}$ be i.i.d. random variables with $X^{(n, \ell)}_j \sim \text{CCP}(n,\ell)$. Then,
    \[
    \Bbb{E} \left[ \max_{j\in[m]} X^{(n,\ell)}_{j} \right] 
    = n \left(\ln\left(nm\right) + \left(\ell-1\right)\ln \ln n+ C_{\ell} \pm o(1) \right)
    \]
    where $C_{\ell} = \gamma - \ln((\ell-1)!)$.
\end{restatable}

\subsection{The MDS \emph{Array} Code Case}

Our next result can be viewed as a generalization of \Cref{thm:mds-version}, motivated by insights from the theory of regenerating codes in classical distributed storage.
Let $S'\in \Sigma^{n\times(M-r)}$ represent the information to be stored. In the DNA-DSS given in \Cref{cor:MDS-DDSS}, we encoded each row of $S'$ independently to obtain the encoded matrix $S$. Consequently, to reconstruct any column of $S$, it was necessary to read at least a fixed number of symbols from every row of $S$.

We now proceed with a more general construction. Partition $S'$ into $n/b$ consecutive submatrices $S'_1, \ldots, S'_{n/b}$, each of dimension $b \times (M-r)$. Concretely, the first submatrix is formed by the top $b$ rows, the second by the next $b$ rows, and so forth.
Each submatrix $S'_i$ is then encoded into a new submatrix $S_i$ of size $b\times M$ using an encoder $\mathcal{E}:\Sigma^{b\times (M-r)} \to \Sigma^{b\times M}$ of an MDS array code. The resulting matrix is obtained by stacking $S_1, \ldots, S_{n/b}$ vertically. Note that here, the MDS property states that knowing the values of any $M-r$ columns of $S$ suffice to recover $S$.

Assume now that the $p$-th container has failed and needs to be reconstructed, and focus, for the moment, only on the first block $S_1$. 
Let $M_{p}\triangleq [M]\setminus \{p\}$ denote the set of all remaining operational containers.  
The encoding function $\mathcal{E}$, together with the failed index $p$, defines a \emph{bad blocks configurations family} $\calB_p \subseteq \mathcal{P}([b]\times M_p)$ (here, $\calP$ denotes the power set) consisting of all index sets that cannot be used to reconstruct the $p$-th column.
Namely, $B \in \calB_p$ if and only if the $p$-th column of $S_1$ cannot be reconstructed from the collection of symbols $\{(S_1)_{i,j} \mid (i,j) \in B\}$.  
Formally,

\begin{restatable}{definition}{defBadConf}
\textup{(Bad blocks configuration with respect to $\mathcal{E}$ and $p$)} \label{def:bad-block-conf}
    Let $\mathcal{E}:\Sigma^{b\times (M-r)} \to \Sigma^{b\times M}$ and fix an index $p \in [M]$. Define $M_p \triangleq [M]\setminus\{p\}$ and let $\calB_p \subseteq \calP([b] \times M_p)$ be a collection of subsets with the following property: a set $B$ belongs to $\calB_p$ if and only if there exist two distinct encoded matrices $S, S' \in \mathcal{E}(\Sigma^{b\times (M-r)})$ such that $S_{i,j} = S'_{i,j}$ for every $(i,j)\in B$, while the $p$-th column of $S$ is \emph{not equal} to the $p$-th column of $S'$. 
    Equivalently, even after observing all entries $S_{i,j}$ with $(i,j)\in B$, one still cannot uniquely determine the $p$-th column of $S$.
Moreover, define 
    \begin{align*}
        \alpha^*_p &\triangleq (M-1)b - \max_{B\in \calB_p} |B| \;, \quad \text{and}\\
        \beta^{*}_p &\triangleq |\{B'\in \calB_p \mid |B'| = \max_{B\in \calB_p}|B|\}|\;.
    \end{align*}
    Namely, $\alpha^*_p$ is the smallest number of missing entries of $S$ for which the $p$-th column cannot be recovered and $\beta^{*}_p$ is the be number of bad recovering sets of maximal size in $\calB_p$.
\end{restatable}

Note that since the same encoder $\mathcal{E}$ is applied to every block, the family $\calB_p$ is the same for all matrices $S_1, \ldots, S_{n/b}$. 
Thus, the question we consider is the expected time until, for each block, we have gathered a set of indices that \emph{does not} belong to $\calB_p$. Similarly to \Cref{thm:mds-version}, the next theorem defines the random process and the stopping criterion that characterize recovery, given a bad blocks configuration. 
\begin{restatable}{theorem}{RegeneratingVersion} \label{thm:regenerating-ex-version}
    Let $A^{(0)} = {\boldsymbol{0}}^{n\times m}, A^{(1)}, A^{(2)}, \ldots$ be a sequence of matrices constructed as follows: for each $t \in \mathbb{N}$, we draw $(v_1, \ldots, v_{m}) \sim \mathrm{Unif}([n]^{m})$ and set
    \[
        A^{(t)} = A^{(t-1)} + \left[\bfe_{v_1}, \ldots, \bfe_{v_{m}}\right] \;.
    \]
    Let $b\in \Bbb{N}$ that divides $n$ and let $\calB \subseteq \calP([b]\times [m])$. As \Cref{def:bad-block-conf}, define 
    \[
    \alpha^* \triangleq mb - \max_{B\in \calB} |B| \qquad \beta^{*} \triangleq |\{B'\in \calB \mid |B'| = \max_{B\in \calB}|B|\}|
    \]
    Define $T$ as the random variable that corresponds to the minimal $t$ such that for all $a\in [n/b]$, we have that
    \[
    \left \lbrace (i,j) \in [b] \times[m] \mid A^{(t)}_{i + (a-1)b,j} \neq 0\right \rbrace \notin \calB\;.
    \]
    Then, 
    \[
    \Bbb{E}[T] \leq \frac{n}{\alpha^*}\ln n + \frac{\beta^*}{b \alpha^*}\cdot n+ o(n)\;,
    \]
\end{restatable}

\Cref{thm:regenerating-ex-version}, combined with \Cref{def:bad-block-conf}, yields a framework for establishing an upper bound on the expected recovery time of each container in a DNA-DSS.
\begin{corollary}
    Let $M,b, r, n\in \Bbb{N}$ where $b$ divides $n$. An injective encoding function $\mathcal{E}:\Sigma^{b\times (M-r)} \times \Sigma^{b \times M}$ yields
    %\hanmao{Can you check the domain and co-domain of $E$? It be $\Sigma^{b\times (M-r)}$. Just numbers right?}
    an $(n, M, n (M-r), |\Sigma|)$ DNA-DSS code such that for every $p\in [M]$, we have
    \[
    \Bbb{E}[T_p] \leq \frac{n}{\alpha^*_p} \ln n + \frac{\beta^*_p}{b \alpha^*_p} n + o(n) \;,
    \]
    where $\beta^*_p$ and $\alpha^*_p$ are defined in \Cref{def:bad-block-conf}.
\end{corollary}

\begin{example} 
    We specialize this framework to one of the simplest examples of an MDS regenerating code \cite[Fig. 1]{DBLP:journals/tit/TamoWB13}, defined over $\Bbb{F}_3$. Let $n$ be an even number. 
    Consider the $(n, 4, 2n, 3)$ DNA-DSS defined by the encoding function $\mathcal{E}: \Bbb{F}_3^{2\times 2} \to \Bbb{F}_3^{2 \times 4}$ defined by
    \[
        \mathcal{E} \left(
    \begin{pmatrix}
        a & b \\
        c & d
    \end{pmatrix} \right) =
    \begin{pmatrix}
        a & b & a+b& a+2d\\
        c & d & c+d& c+b
    \end{pmatrix}\;,
    \]
    where arithmetic oprations are done over $\Bbb{F}_3$.
    Assume that $p=1$ and  observe that there are exactly two sets of size $3$ that recover the first column: $\{(1,2), (1,3),(2,4)\}$ and $\{(2,2),(2,3),(1,4)\}$. Further, there are only two sets of size $4$ that do \emph{not} recover the first column: $\{(1,2),(1,3),(1,4), (2,2)\}$ and $\{(1,2),(2,2), (2,3),(2,4)\}$. Thus,
    \[
    \begin{tabular}{c||c|c|c|c|c|c|c|}
         $\delta$ & $0$ & $1$ & $2$ & $3$ & $4$ &$5$ & $6$ \\
         \hline
         $b_{\delta}$ & $1$ & $6$ & $\binom{6}{2}=15$ & $\binom{6}{3} - 2 = 18$ & $2$ & $0 $& $0 $
    \end{tabular}
    \]
    % \[
    % \begin{tabular}{c|c}
    %      $\delta$ & $|\{ B\in \calB \mid |B| = \delta\}|$ \\
    %      \hline \hline
    %      $1$ & $6$ \\
    %      \hline
    %      $2$ & $\binom{6}{2}=15$ \\
    %      \hline
    %      $3$ & $\binom{6}{3} - 2 = 18$ 
    % \end{tabular}
    % \qquad
    % \begin{tabular}{c|c}
    %      $\delta$ & $|\{ B\in \calB \mid |B| = \delta\}|$ \\
    %      \hline \hline
    %      $4$ & $2$  \\
    %      \hline
    %      $5$ & $0$ \\
    %      \hline
    %      $6$ & $0$ \\
    %      \hline
    % \end{tabular}
    % \]
    where $b_{\delta} = |\{B\in \calB_1 \mid |B| = \delta\}| $. 
    Since $M = 4$ and $b = 2$, we get, $\alpha^*_1 = 3\cdot 2 - 4 = 2$ and $\beta^*_1 = 2$. In total, we get
    \[
    \Bbb{E}[T_1]\leq \frac{n}{2} \ln n + \frac{n}{2} + o(n) \;.
    \]
    Note that, setting $M = 4$ and $r = 2$ in the scalar MDS DNA-DSS construction of \Cref{cor:MDS-DDSS}, we obtain $\Bbb{E}[T_1] \approx \frac{n}{2} \ln n + 0.838 n \pm o(n)$, which shows that the regenerating code provides an improvement of at least $0.338n$.
\end{example} 
\section{Preliminaries}
    We shall use the following limit theorem about the coupon collector distribution.
    \begin{theorem}[Theorem 5.13 \cite{bookProbabilityAndComputing}] \label{thm:gambel-conv-in-dist}
        Let $X\sim\text{CCP}(n)$. Then, for any constant $c$,
        \begin{displaymath}
            \lim_{n\to\infty}\Pr\left[X>n\ln n +cn\right] = 1-e^{-e^{-c}}
        \end{displaymath}
    \end{theorem}
    This limit theorem is exactly a special case of the Gumbel distribution whose definition is given next.
    \begin{definition}[Gumbel distribution]
    \label{def:gumbel-dist}
        The cumulative distribution function of Gumbel distribution with parameters $\mu,\beta$ is 
        \begin{displaymath}
            F(x;\mu,\beta) = e^{-e^{\frac{-(x-\mu)}{\beta}}}.
        \end{displaymath}
        If $X$ is a random variable with Gumbel distribution with parameters $\mu,\beta$, we denote $X\sim \text{Gumbel}(\mu, \beta)$.
    \end{definition}

    \begin{definition}[Convergence in distribution]
    \label{def:conv-in-dist}
        A sequence $\calX = X^{(1)},X^{(2)},\dots$ of real-valued random variables with cumulative distribution functions $F_1,F_2,\dots$ is said to converge in distribution to a random variable $X$ with cumulative distribution function $F$ if $\lim_{n\to \infty}F_n(x) = F(x)$, for every $x\in \mathbb{R}$ at which $F$ is continuous.
        We denote convergence in distribution by $\mathcal{X}\overset{\mathcal{D}}{\to } X$.
    \end{definition}

    \begin{lemma}\label{lem:lim-of-max-conv-in-dist}
        Let $\calX_1, \ldots, \calX_m$ be $m$ sequences of random variables such that for every $n\in \Bbb{N}$, $X_1^{(n)},X_2^{(n)},\ldots,X_m^{(n)}$ are i.i.d. and for all $j\in [m]$, we have $\calX_j \overset{\mathcal{D}}{\to }\text{Gumbel}(\mu, \beta)$.
        For every $n\in \Bbb{N}$, denote $X_{\max}^{(n)}\triangleq \max_jX_j^{(n)}$ and let $\calX_{\max} \triangleq X_{\max}^{(1)}, X_{\max}^{(2)}, \ldots$. It holds that 
        \begin{align*}
            \calX_{\max}\overset{\mathcal{D}}{\to } \Gumbel(\beta\ln(m) + \mu, \beta).
        \end{align*}
    \end{lemma}

    %  \begin{lemma}\label{lem:lim-of-max-conv-in-dist} 
    %     Let $X_1^{(n)},X_2^{(n)},\ldots,X_m^{(n)}$ i.i.d. random variables such that for any $i$, $X_i^{(n)}\overset{\mathcal{D}}{\to }\text{Gumbel}(\mu, \beta)$. \al{do we want to change the notations or leave it with $X_i^{(n)}$?}
    %     Denote $X_{\max}^{(n)}\triangleq \max_iX_i^{(n)}$. Then,
    %     \begin{align*}
    %         X_{\max}^{(n)}\overset{\mathcal{D}}{\to } \text{Gumbel}(\beta\ln(m) + \mu, \beta).
    %     \end{align*}
    % \end{lemma}
    \begin{proof}
        By the definition of convergence in distribution (\Cref{def:conv-in-dist}) and the Gumbel distribution (\Cref{def:gumbel-dist}), for every $t$ and $j\in [m]$,
        \begin{displaymath}
            \lim_{n\to\infty}\Pr\left[X_j^{(n)}\leq t\right] = e^{-e^{-\frac{(t-\mu)}{\beta}}}.
        \end{displaymath}
        For every fixed $n\in \Bbb{N}$, since $X_1^{(n)}, \ldots, X_m^{(n)}$ are i.i.d., for every fixed $t$ and  $m$, we have 
        \begin{displaymath}
            \Pr\left[X_{\max}^{(n)}\leq t\right] = \Pr\left[\forall j\in [m], X_{j}^{(n)}\leq t\right] = \Pr\left[X_1^{(n)}\leq t\right]^m.
        \end{displaymath} Therefore, for any fixed $t,m$, 
        \begin{displaymath}
            \lim_{n\to \infty} \Pr\left[X_{\max}^{(n)}\leq t\right] 
            = \lim_{n\to \infty}\Pr\left[X_1^{(n)}\leq t\right]^m 
            = \exp\left(-e^{-\frac{(t-\mu)}{\beta}}\right)^m 
            = \exp\left(-e^{-\frac{(t-(\mu +\beta\ln(m))}{\beta}}\right).
        \end{displaymath} 
        By \Cref{def:gumbel-dist} and \Cref{def:conv-in-dist}, we conclude that $\calX_{\max}\overset{\mathcal{D}}{\to } \text{Gumbel}(\beta\ln(m)+\mu, \beta)$ as required. 
    \end{proof}

    % \begin{lemma}[Chebyshev's inequality]
    %         Let $X$ be a random variable with finite non-zero variance $\sigma^2$ (and thus finite expected value $\mu$). Then for any real number $k>0$, $\Pr[|X-\mu|\geq k] \leq \frac{\sigma^2}{k^2}$.
    %         \label{lem:cheb-ineq}
    % \end{lemma}

    For the following definition we shall need the following notation of an indicator. Let $E$ be any event. Then, define
    
    \[
    \mathds{1}_{E} = \begin{cases}
            1 & \text{if } E \text{ occurs} \\
            0 & \text{otherwise}
        \end{cases}
    \]
        
    \begin{definition}[Uniformly integrable, Definition 4.1 \cite{gut2006probability}]
        A sequence $\calX = X^{(1)},X^{(2)},\ldots$ of random variables with finite expectation is called uniformly integrable if for any $n\in \Bbb{N}$, it holds that
        $\mathbb{E}[|X^{(n)}|\cdot \mathds{1}_{|X^{(n)}|> K}]\to 0 \text{ as } K\to \infty$, uniformly in $n$.
    \label{def:unif-integrable}
    \end{definition}

    For a sequence of uniformly integrable variables, we have the following.
       \begin{theorem}[Theorem 3.5, \cite{billingsley2013convergence}]\label{thm:conv-of-expec}
        Let $\calX = X^{(1)}, X^{(2)},\dots$ be a uniformly integrable sequence such that $\calX \overset{\mathcal{D}}{\to} X$ to some random variable $X$. Then $X$ has finite expectation and $\lim_{n\to \infty}\mathbb{E}[X^{(n)}] = \mathbb{E}[X]$.
    \end{theorem}

    Throughout the paper, we shall want to prove that a sequence of random variables is uniformly integrable. The following lemma gives a sufficient condition for that.
    \begin{lemma}\label{lem:exp-tail-impl-unif-conv}
        Let $\calX = X^{(1)},X^{(2)},\dots$ be a sequence of random variables and assume there exist $C,c>0$ such that for every $t\geq 0$ and any $n\in \Bbb{N}$, it holds that $\Pr[|X^{(n)}|> t] \leq Ce^{-ct}$. Then, $\calX$ is uniformly integrable.
    \end{lemma}
    \begin{proof}
        We will verify that the two conditions specified in \Cref{def:unif-integrable} hold.
        \begin{enumerate}
            \item For every $n\in \Bbb{N}$, we have that
            \begin{align*}
                \mathbb{E}[|X^{(n)}|] 
                &= \int_0^\infty \Pr\left[|X^{(n)}|> t\right]dt \\
                &\leq C\int _0^\infty e^{-ct}dt = \frac{C}{c} < \infty\;.
            \end{align*} 
            Thus, $X^{(n)}$ has finite expectation.
            \item For any $n\in \Bbb{N}$ and positive $K$,
            \begin{align*}
                0\leq \mathbb{E}[|X^{(n)}|\cdot \mathds{1}_{|X^{(n)}|\geq K}] 
                &= \int_K^\infty \Pr\left[|X^{(n)}|> t\right]dt \\
                &\leq C\int _K^\infty e^{-ct}dt 
                = \frac{C}{c}e^{-cK} \;.
            \end{align*}
            Since $\lim_{K\to\infty}\frac{C}{c}e^{-cK}=0$ and is independent of $n$, $\lim_{K\to\infty}\mathbb{E}[|X^{(n)}|\cdot \mathds{1}_{|X^{(n)}|\geq K}]= 0$ uniformly in $n$.
        \end{enumerate}
         We conclude that $\calX$ is uniformly integrable.
    \end{proof}

\section{Warmup: Proof of \Cref{thm:single-copy-no-redun}}\label{sec:proof-of-one-redund}
    We start by introducing the notations of this section. Let $m$ be a constant positive integer. For all $j\in [m]$, let $\calX_j = \{X_j^{(n)}\}_{n\geq 2} $ be a sequence of random variables where $X_j^{(n)} \sim \CCP(n)$. 
    Furthermore, we assume that for every fixed $n\in \Bbb{N}$, the variables $X_1^{(n)}, \ldots, X_m^{(n)}$ are i.i.d. 
    For every $n\in \Bbb{N}$, let $X_{\max}^{(n)}\triangleq \max_{j\in [m]} X_j^{(n)}$.

    We will be interested in the following normalized versions of $\calX_1, \ldots,\calX_m$.
    For every $j\in [m]$ and $n\in \Bbb{N}$, define $Y_j^{(n)} \triangleq \frac{X_j^{(n)}}{n}-\ln n$ and $Y_{\max}^{(n)}\triangleq \max_{j\in [m]}Y_j^{(n)}$. Finally, for every $j\in [m]$, we set $\calY_j \triangleq \{Y_j^{(n)}\}_{n\geq 2}$ and  $\calY_{\max} \triangleq \{Y_{\max}^{(n)}\}_{n\geq 2}$.

    The proof of \Cref{thm:single-copy-no-redun} consists of the following steps. First, in \Cref{clm:y-max-conv-in-dist}, we show that $\calY_{\max}$ converges in distribution to the Gumbel distribution with $\mu = \ln m$ and $\beta = 1$. Then, in \Cref{clm:y-max-bound}, we bound the tails of $Y_{\max}^{(n)}$ which is needed to show in \Cref{clm:y-max-unif-int}, that $\calY_{\max}$ is uniformly integrable. Finally, to prove \Cref{thm:single-copy-no-redun} we will apply \Cref{thm:conv-of-expec}.

   \begin{claim}\label{clm:y-max-conv-in-dist}
        It holds that
        \begin{align*}
            \calY_{\max}\overset{\mathcal{D}}{\to } \Gumbel(\ln(m), 1).
        \end{align*}
    \end{claim}
    \begin{proof}
        Note that for all $n\in \Bbb{N}$, and every $j\in [m]$ we have that $nY_j^{(n)} + n\ln n \sim \CCP(n)$. By \Cref{thm:gambel-conv-in-dist}, 
        \[
        \lim_{n\to \infty}\Pr\left[nY_j^{(n)} + n\ln n > n \ln n + cn\right] = 1 - {e^{-e}}^{-c} \;,
        \]
        which implies that 
        \[
        \lim_{n\to \infty}\Pr[Y_j^{(n)} \leq c] = {e^{-e}}^{-c} \;,
        \]
        and thus, by \Cref{def:conv-in-dist}, the sequence $\calY_j = Y_j^{(1)}, Y_j^{(2)},\ldots$ converges in distribution to $\Gumbel(0,1)$. 
        Applying \Cref{lem:lim-of-max-conv-in-dist} with $\calY_1,\ldots,\calY_m$ proves the claim.
        % XXXXX\\\\
        % This is a direct result of \Cref{thm:gambel-conv-in-dist} and \Cref{lem:lim-of-max-conv-in-dist}.
    \end{proof}

\begin{claim} \label{clm:y-max-bound}
    For any $n\geq 2$ and any $x\geq 0$, 
    \[
    \Pr\left[ \left|Y_{\max}^{(n)}\right| > x \right] \leq (2 + m) \cdot e^{-x}.
    \]
\end{claim}

\begin{proof}
    We first show that $\Pr[ Y_{\max}^{(n)} > x] \leq m\cdot e^{-x}$. Indeed, 
    By the definition of $Y_{\max}^{(n)}$, we have 
    \begin{align*}
            \Pr\left[Y_{\max}^{(n)}> x\right] 
            &=  \Pr\left[X_{\max}^{(n)}> n(\ln n  + x)\right] \\
            &= \Pr\left[\text{some coupon was \textbf{not} collected at time }n(\ln n  + x)\right] \\
            &\leq nm\left(1-\frac{1}{n}\right)^{n(\ln n  + x)} \\
            &\leq me^{-x} \;,
    \end{align*}
    where the first inequality is by a union bound.

    Now, we show that $\Pr[ Y_{\max}^{(n)} < -x] \leq 2e^{-x}$. First observe that 
    \[
    \Pr\left[Y_{\max}^{(n)}<- x\right] \leq \Pr\left[Y_{\max}^{(n)}\leq- x\right] =\Pr\left[Y_1^{(n)}\leq- x\right]^m = \Pr\left[X_1^{(n)} \leq n(\ln n  - x)\right]^m\;.
    \]
    Next, we shall bound $\Pr[X_1^{(n)} < n(\ln n  - x)]$.
    Fix $x$, set $t\triangleq n(\ln n  - x)$, and define $W^{(t)}$ to be the number of missing coupons after $t$ draws were drawn for the first collector. Note that $W^{(t)}=\sum_{j=1}^n I^{(t)}_j$ where $I^{(t)}_j$ is the indicator that $j$-th coupon was \emph{not} collected in the first $t$ draws. 
    Clearly, 
    \[\Pr\left[X_1^{(n)} \leq n(\ln n  - x)\right] = \Pr\left[X_1^{(n)} \leq t\right] = \Pr\left[W^{(t)}=0\right]
    \;.\]
    We shall bound the expectation and variance of $W^{(t)}$ in order to upper bound $\Pr[W^{(t)}=0]$.
    We have that

    \[
    \mathbb{E}\left[W^{(t)}\right] = \sum_{j=1}^n\mathbb{E}\left[I_j^{(t)}\right] = n\cdot \left(1 - \frac{1}{n}\right)^{n (\ln n - x)} \geq n\cdot e^{-\frac{n}{n-1} \cdot (\ln n - x)} \geq  n^{-\frac{1}{n-1}} \cdot e^x \geq \frac{1}{2}e^x \;,
    \]
    where the first inequality follows since $(1 - 1/n)^n \geq e^{-{n/(n-1)}}$ for all $n\geq 2$ and last inequality follows since $n^{-1/(n-1)}$ is monotonically increasing. For the variance of $W^{(t)}$, we first observe that
    \begin{align*}
        \Var\left(W^{(t)}\right)
        &= \sum_{i\in [n]}\Var\left(I_i^{(t)}\right) + 2\sum_{1 \leq i< j\leq n} \Cov\left(I_i^{(t)}, I_j^{(t)}\right) \;.
        % \\
        % &= \sum_{i\in [n]} \Bbb{E}[(I_i^{(t)})^2] - \Bbb{E}[(I_i^{(t)})]^2 + 2\sum_{1 \leq i< j\in n} \Bbb{E}[I_i^{(t)}\cdot I_j^{(t)}] - \Bbb{E}[I_i^{(t)}] \cdot \Bbb{E}[I_j^{(t)}]\\
        % &\leq \sum_{i\in [n]} \Bbb{E}[I_i^{(t)}] + 2 \cdot \binom{n}{2} \cdot \left( \left( 1 - \frac{2}{n}\right)^T - \left( 1 - \frac{1}{n}\right)^{2T}\right) \\
        % &\leq \Bbb{E}[W^{(t)}] \;,
    \end{align*}
    Now, since $I_i^{(t)}$ is an indicator, then $\Var(I_i^{(t)}) \leq \Bbb{E}[I_i^{(t)}]$.
    Moreover, 
    \[
    \Cov\left(I_i^{(t)}, I_j^{(t)}\right) = \Bbb{E}\left[I_i^{(t)}\cdot I_j^{(t)}\right] - \Bbb{E}\left[I_i^{(t)}\right] \cdot \Bbb{E}\left[I_j^{(t)}\right] = \left( 1 - \frac{2}{n}\right)^t - \left( 1 - \frac{1}{n}\right)^{2t} <0,
    \]
    where the inequality follows since $(1 - 2/n) < (1 - 1/n)^2$. Combining, we have that $\Var(W^{(t)}) < \Bbb{E}[W^{(t)}]$.

    Thus, by plugging these bounds into the Chebyshev inequality,
    \begin{align*}
            \Pr\left[W^{(t)}=0\right] \leq \Pr\left[\left|W^{(t)}- \mathbb{E}\left[W^{(t)}\right]\right| \geq \mathbb{E}\left[W^{(t)}\right]\right] \leq \frac{\Var\left(W^{(t)}\right)}{\mathbb{E}\left[W^{(t)}\right]^2} \leq \frac{1}{\mathbb{E}\left[W^{(t)}\right]}\leq 2 e^{-x}.
        \end{align*}
    We conclude that
    \begin{align*}
            \Pr\left[ Y_{\max}^{(n)} < -x \right] = \left(\Pr\left[W^{(t)}=0\right] \right)^m 
            \leq \Pr\left[W^{(t)}=0\right]
            \leq 2e^{-x} \;. 
    \end{align*}
\end{proof}

\begin{corollary}
    \label{clm:y-max-unif-int}
    The sequence $\calY_{\max}$ is uniformly integrable.
\end{corollary}
\begin{proof}
    Follows directly from \Cref{lem:exp-tail-impl-unif-conv} and \Cref{clm:y-max-bound}.
\end{proof}

We are now ready to prove \Cref{thm:single-copy-no-redun} which is restated for convenience.

\singleCopyNoRedund*
    \begin{proof}
        According to \Cref{clm:y-max-unif-int}, the sequence $\calY_{\max}$ is uniformly integrable, and converges in distribution to $\text{Gumbel}(\ln (m), 1)$ and so by \Cref{thm:conv-of-expec}, we have that 
        \begin{align*}
            \lim_{n\to\infty}\mathbb{E}\left[Y_{\max}^{(n)}\right] 
            = \mathbb{E}\left[\text{Gumbel}(\ln(m), 1)\right] = \ln(m) + \gamma \;,
        \end{align*}
        Therefore, we can write 
        \[
        \mathbb{E}\left[Y_{\max}^{(n)}\right] = \ln(m) + \gamma \pm o(1) \;.
        \]
        Remembering that $X_{\max}^{(n)} = n\ln n + nY_{\max}^{(n)}$, we conclude that
        \[
        \mathbb{E}[X_{\max}^{(n)}] 
            = n\ln n + n\mathbb{E}[Y_{\max}^{(n)}] 
            = n\left(\ln(mn) + \gamma  \pm o(1)\right) \;.
        \]
    \end{proof}

\section{Proof of \Cref{thm:multi-copy-no-redun}}
    The proof of  \Cref{thm:multi-copy-no-redun} follows the same general strategy as that of \Cref{thm:single-copy-no-redun}. Specifically, we combine convergence in distribution with uniform integrability to establish the convergence of expectations of a sequence of normalized random variables and then infer the expectation of the target random variable. However, the analysis in this case is more involved.  

    We shall first establish notations. Fix $\ell\in \mathbb{N}$. For every $j\in [m]$, let $\calX_j = \{X_j^{(n,\ell)}\}_{n\geq 2}$ be a sequence of random variables where $X_j^{(n,\ell)} \sim \CCP(n, \ell)$. As before, we assume that for any fixed $n\in \Bbb{N}$, we have that $X_1^{(n,\ell)},X_2^{(n,\ell)},\ldots,X_m^{(n,\ell)}$ are i.i.d., and we also denote $X_{\max}^{(n,\ell)} \triangleq \max_{j\in [m]}X_j^{(n,\ell)}$.
    For every $j\in [m]$, define $Y_j^{(n,\ell)} \triangleq \frac{X_j^{(n,\ell)}}{n}-\ln n-(\ell-1)\ln \ln n$ and, as before, $\calY_j \triangleq \{Y_j^{(n,\ell)} \}_{n\geq 2}$. Finally, define $\calY_{\max} = \{Y_{\max}^{(n,\ell)} \}_{n\geq 2}$ where $Y_{\max}^{(n,\ell)} = \max_{j\in [m]}Y_j^{(n,\ell)}$. Recall that we assume that $m$ and $\ell$ are constants with respect to $n$.

    % XXXXXX\\\\\\
    % For the rest of the proof let $X_1^{(n,\ell)},X_2^{(n,\ell)},\ldots,X_m^{(n,\ell)}$ be i.i.d. random variables where $X_i^{(n,\ell)}\sim \text{CCP}(n,\ell)$ and for every $i\in [n]$, define $Y_i^{(n,\ell)} \triangleq \frac{X_i^{(n,\ell)}}{n}-\ln n-(\ell-1)\ln \ln n$. \al{for some reason it bothers me we sometimes use $\ln n $ and sometimes $\ln n$, which one is better?}
    % Denote $X_{\max}^{(n,\ell)}\triangleq \max_{i\in [m]} X_i^{(n,\ell)}$ and $Y_{\max}^{(n,\ell)}\triangleq \max_{i\in [m]}Y_i^{(n,\ell)}$.

    The next theorem, due to Erd\H{o}s and R\'enyi \cite{erdHos1961classical}, extends the limit theorem for the classical coupon collector problem (given in \Cref{thm:gambel-conv-in-dist}), to the version of the coupon collector problem where one seeks $\ell$ copies of each coupon.
    \begin{theorem}[\cite{erdHos1961classical} ]\label{thm:ccp-l-set-dist}
        Let $X^{(n,\ell)}\sim \text{CCP}(n,\ell)$, then 
        \begin{displaymath}
            \lim_{n\to\infty} \Pr\left[\frac{X^{(n,\ell)}}{n}<\ln n + (\ell-1)\ln\ln n +x \right] = \exp\left(-\frac{e^{-x}}{(\ell-1)!}\right)\;.
        \end{displaymath}
        Equivalently, for  $Y^{(n,\ell)}\triangleq \frac{X^{(n,\ell)}}{n} -\ln n - (\ell-1)\ln\ln n$, we have that $Y^{(n,\ell)}\overset{\mathcal{D}}{\to } \text{Gumbel}(-\ln((\ell-1)!), 1)$\,.
    \end{theorem}
    \Cref{thm:ccp-l-set-dist}, alongside \Cref{lem:lim-of-max-conv-in-dist}, lead to the following corollary. 
    \begin{corollary}\label{cor:multi-set-max-conv-in-dist}
        It holds that
        \begin{align*}
            \calY_{\max}\overset{\mathcal{D}}{\to } \Gumbel\left(\ln\left(\frac{m}{(\ell-1)!}\right), 1\right).
        \end{align*}
    \end{corollary}
    \begin{proof}
        Note that for all $n\in \Bbb{N}$, and every $j\in [m]$ we have that $nY_j^{(n,\ell)} + n\ln n + n\left(\ell-1\right)\ln \ln n \sim \CCP(n,\ell)$. By \Cref{thm:ccp-l-set-dist} and \Cref{def:conv-in-dist}, for any $j\in [m]$, $\calY_j\overset{\mathcal{D}}{\to } \text{Gumbel}(-\ln((\ell-1)!), 1)$.  Applying \Cref{lem:lim-of-max-conv-in-dist} with $\calY_1,\ldots,\calY_m$ proves the claim.
    \end{proof}

    As in the previous section, in order to prove that $ \calY_{\max}$ is uniformly integrable, we will show that $\Pr\left[|Y_{\max}^{(n, \ell)}| >x\right]$ can be exponentially bounded and that the bound does not depend on $n$.
     \begin{claim}\label{clm:l-sets-right-tail}
        There exists $C_0=C_0(\ell,m)$, such that for any $x\geq0$,
        \begin{align*}
            \Pr\left[Y_{\max}^{(n, \ell)} >x\right] \leq C_0 e^{-\frac{x}{2}}.
        \end{align*}
    \end{claim}
    \begin{proof}
        Denote  $t \triangleq n(\ln n + (\ell-1)\ln\ln n + x)$. We have that 
        \begin{align*}
            \Pr\left[Y_{\max}^{(n, \ell)}> x\right] 
            &= \Pr\left[X_{\max}^{(n, \ell)}> t \right] \\
            &= \Pr\left[\text{some coupon was collected \textbf{less} than $\ell$ times at time } t\right] \\
            &\leq nm\cdot \Pr\left[\text{a specific coupon was collected \textbf{less} than $\ell$ times at time } t\right] \\
            &= nm\cdot \Pr\left[\text{Bin}\left(t, \frac{1}{n}\right) \leq \ell-1\right] \\
            &= nm\cdot \sum_{r=0}^{\ell-1}\binom{t}{r}\left(\frac{1}{n}\right)^r\left(1-\frac{1}{n}\right)^{t-r} \\
            &= nm\left(1-\frac{1}{n}\right)^{t}\cdot \sum_{r=0}^{\ell-1}\binom{t}{r}\left(\frac{1}{n}\right)^r\left(1-\frac{1}{n}\right)^{-r} \\
            &= nm\left(1-\frac{1}{n}\right)^{t}\cdot \sum_{r=0}^{\ell-1}\binom{t}{r}\left(\frac{1}{n-1}\right)^r\\
            &\leq nm\left(1-\frac{1}{n}\right)^{t}\cdot \sum_{r=0}^{\ell-1}\left(\frac{t}{n-1}\right)^r \\
            %&\leq nm\left(1-\frac{1}{n}\right)^{t}\cdot 2^\ell \ell \left(\frac{t}{n}\right)^{\ell-1} \\
            &\leq 2^\ell m \ell \cdot n e^{-\frac{t}{n}}\left(\frac{t}{n}\right)^{\ell-1}\;,
        \end{align*}
        where the first inequality is obtained via a union bound, the second inequality applies the bound $\binom{a}{b} < a^b$, and the third inequality uses the bound $(1 - 1/n)^n \leq e^{-1}$ together with the fact that the final term in the sum is the largest one.
        For any $\ell$, there exists $x_0(\ell)\geq \ell$ such that for any $x>x_0(\ell)$, $x^{\ell-1}\leq e^\frac{x}{2}$. Hence, for any $x>x_0(\ell)$,
        % Note that
        \begin{align*}
            n e^{-\frac{t}{n}}\left(\frac{t}{n}\right)^{\ell-1}
            %= n e^{-(\ln n + (\ell-1)\ln\ln n + x)}\left(\frac{t}{n}\right)^{\ell-1} 
            = e^{-x}\left(\frac{t}{n\ln n}\right)^{\ell-1} 
            \leq e^{-x}\left(\ell+2x\right)^{\ell-1}
            % \leq e^{-x}(3x)^{\ell-1} 
            \leq 3^{\ell-1}e^{-\frac{x}{2}},
            % \leq e^{-\Omega_{\ell}(x)} \;,
        \end{align*}
        where the first inequality holds since $\frac{t}{n\ln n} = 1 + \frac{(\ell-1)\ln\ln n}{\ln n}  + \frac{x}{\ln n} < \ell + 2x$, and the second inequality holds by the requirements on $x_0(\ell)$. 
        Put it all together,
        %for any $x>x_0(\ell)$,
        \begin{align*}
            \Pr\left[Y_{\max}^{(n, \ell)}> x\right] 
            \leq 2^\ell m \ell \cdot ne^{-\frac{t}{n}}\left(\frac{t}{n}\right)^{\ell-1}
            \leq 2^\ell 3^{\ell-1} m \ell e^{-\frac{x}{2}}.
            % \leq 2^{\ell}m \ell e^{-\Omega_{\ell}(x)} = e^{-\Omega_{m,\ell} (x)} \;.
        \end{align*} 
        In addition, for any $0\leq x\leq x_0(\ell)$, $\Pr[Y_{\max}^{(n, \ell)}> x]\leq 1 \leq e^{x_0(\ell)}e^{-\frac{x}{2}}$. Denote $C_0(\ell,m) \triangleq \max\left\{e^{x_0(\ell)}, 3^{2\ell-1} m \ell\right\}$. Then, $\Pr[Y_{\max}^{(n, \ell)}> x]  \leq C_0(\ell,m) e^{-\frac{x}{2}}$ for any $x\geq 0$ as required.
    \end{proof}

    To bound the left tail, namely, to bound $\Pr\left[Y_{\max}^{(n, \ell)}<- x\right]$, we will need the notion of negatively associated random variables.
    \begin{definition}[Negative association] \label{def:neg-assoc}
        A collection $X_1,X_2,\dots,X_n$ of random variables is called \emph{negatively associated} (NA) if for every two disjoint index sets, $I,J\subseteq [n]$, and every $f:\mathbb{R}^{|I|}\to \mathbb{R}$ and $g:\mathbb{R}^{|J|}\to \mathbb{R}$ that are both non-decreasing or both non-increasing, we have
        \begin{displaymath}
            \mathbb{E}\left[f\left(X_i, i\in I\right)g\left(X_j,j\in J\right)\right]\leq \mathbb{E}\left[f\left(X_i, i\in I\right)\right] \mathbb{E}\left[g\left(X_j,j\in J\right)\right] \;.
        \end{displaymath}
        Equivalently, $\Cov(f(X_i, i\in I), g(X_j,j\in J))\leq 0$.
    \end{definition}
    
    \begin{theorem}[Theorem 13, \cite{dubhashi1996balls}]\label{thm:neg-assoc}
        Suppose we throw $m$ balls into $n$ bins independently at random. For $i\in [n]$, let $B_i$ be the random variable denoting the number of balls in the $i$th bin. 
        Then, the variables $B_1, \ldots, B_n$ are NA.
    \end{theorem}

    \begin{claim}
    \label{clm:l-sets-left-tail}
        There exists $C_1=C_1(\ell)$, such that for any $x\geq 0$,
        \begin{align*}
            \Pr\left[Y_{\max}^{(n, \ell)}\leq - x\right] \leq C_1e^{-\frac{x}{2}}.
        \end{align*}
    \end{claim}
    \begin{proof}
        Similarly to the proof of \Cref{clm:y-max-bound}, 
        \begin{align*}
            \Pr\left[Y_{\max}^{(n,\ell)}\leq- x\right] =\Pr\left[Y_1^{(n, \ell)}\leq- x\right]^m
            \leq \Pr\left[Y_1^{(n, \ell)}\leq -x\right]
            = \Pr\left[X_1^{(n,\ell)} \leq n\left(\ln n + \left(\ell-1\right)\ln\ln n  - x\right)\right].
        \end{align*}
        Fix $x$, and define $t\triangleq n(\ln n + (\ell-1)\ln\ln n  - x)$. If $x> \ln n + (\ell-1)\ln\ln n -\ell$, then $t < n\ell$. Hence $\Pr[X_1^{(n, \ell)}<t] = 0$, and in particular $\Pr[Y_{\max}^{(n, \ell)}<- x] = 0 \leq Ce^{-cx}$ for any $C,c>0$. Therefore, for the rest of the proof, we assume that $0\leq x\leq \ln n + (\ell-1)\ln\ln n -\ell$.

        Let $B_j^{(t)}$ be the number of copies of coupon $j$ after the first $t$ draws of one collector, and define $I^{(t)}_j \triangleq \mathds{1}_{B_j^{(t)} <\ell}$. Then $I^{(t)}_j$ is the indicator that coupon $j$ has less than $\ell$ copies after the first $t$ draws, and $W^{(t)}\triangleq \sum_{j=1}^n I^{(t)}_j$ is the number of such coupons. Hence,
        \begin{align*}
            \Pr\left[X_1^{(n, \ell)} \leq n\left(\ln n + (\ell-1)\ln\ln n - x\right)\right] 
            = \Pr\left[X_1^{(n, \ell)} \leq t\right] 
            = \Pr\left[W^{(t)}=0\right].
        \end{align*}
        Analogously to the proof of \Cref{clm:y-max-bound}, we now derive bounds on the expectation and variance of $W^{(t)}$ in order to obtain an upper bound on $\Pr[W^{(t)}=0]$.
        Note that the variables $B_1^{(t)}, \ldots, B_n^{(t)}$ correspond exactly to the balls and bins paradigm (with $t$ balls and $n$ bins), and thus, they are NA according to \Cref{thm:neg-assoc}. Moreover, observe that the indicator function $f(B)=\mathds{1}_{B <\ell}$ is non-increasing, thus for any $i\ne j$, $\Cov(I^{(t)}_i,I^{(t)}_j) \leq 0$. 
        Therefore,
        %By \Cref{thm:neg-assoc}, $\left\{B_j\right\}$ are NA and the indicator function $f(B)=\mathds{1}_{B <\ell}$ is non-increasing, thus for any $i\ne j$, $\Cov(I^{(t)}_i,I^{(t)}_j) \leq 0$, and
        \begin{align*}
            \Var\left(W^{(t)}\right) 
            &= \sum_{j\in [n]} \Var\left(I_j^{(t)}\right) + 2\sum_{1\leq j<i\leq n} \Cov\left(I_j^{(t)},I_i^{(t)}\right) \\
            & \leq \sum_{j\in [n]} \Var\left(I_j^{(t)}\right)
            \leq \sum_{j\in [n]}\mathbb{E}\left[I_j^{(t)}\right] 
            = \mathbb{E}\left[W^{(t)}\right],
        \end{align*}
        where the last inequality holds again since for any indicator $I$, $\Var(I)\leq \mathbb{E}[I]$.
        Moreover, 
        \begin{align*}
            \mathbb{E}\left[W^{(t)}\right] &= \sum_{j\in [n]}\mathbb{E}\left[I_j^{(t)}\right] = n\sum_{r=0}^{\ell-1}\binom{t}{r}\left(\frac{1}{n}\right)^r\left(1-\frac{1}{n}\right)^{t-r} > 0.
        \end{align*}
        Hence, by the Chebyshev inequality,
        \begin{align*}
            \Pr\left[W^{(t)}=0\right] \leq \Pr\left[\left|W^{(t)}- \mathbb{E}\left[W^{(t)}\right]\right|\geq \mathbb{E}\left[W^{(t)}\right]\right] \leq \frac{\Var\left(W^{(t)}\right)}{\mathbb{E}\left[W^{(t)}\right]^2} \leq \frac{1}{\mathbb{E}\left[W^{(t)}\right]}.
        \end{align*}
        
        In addition, 
        \begin{align*}
            \mathbb{E}\left[W^{(t)}\right] 
            &= n\sum_{r=0}^{\ell-1}\binom{t}{r}\left(\frac{1}{n}\right)^r\left(1-\frac{1}{n}\right)^{t-r} \\
            &= n\left(1-\frac{1}{n}\right)^{t}\sum_{r=0}^{\ell-1}\binom{t}{r}\left(\frac{1}{n-1}\right)^r \\
            &\geq ne^{-\frac{t}{n-1}}\sum_{r=0}^{\ell-1}\binom{t}{r}\left(\frac{1}{n-1}\right)^r.
        \end{align*}
        Divide into cases: 
        \begin{enumerate}
            \item $2(\ell-1)\ln \ln n \leq x$:
            \begin{align*}
                ne^{-\frac{t}{n-1}}\sum_{r=0}^{\ell-1}\binom{t}{r}\left(\frac{1}{n-1}\right)^r 
                &\geq ne^{-\frac{t}{n-1}} 
                % = ne^{-\frac{n}{n-1}(\ln n + (\ell-1)\ln\ln n - x)} \\
                = n^{-\frac{1}{n-1}}e^{\frac{n}{n-1}( x-(\ell-1)\ln\ln n)} 
                \geq \frac{1}{2}e^{\frac{n}{n-1}\frac{x}{2}}
                \geq \frac{1}{2}e^{\frac{x}{2}},
            \end{align*}
            where the second inequality holds by the requirements on $x$ and the monotonicity of $n^{-1/(n-1)}$.
            \item $0\leq x \leq 2(\ell-1)\ln \ln n $: Let $n_1=n_1(\ell)$ be such that for any $ n>n_1$ :  
            \begin{enumerate}
                % \item $n\ln \ln n \geq 1 $
                \item \label{assump:a} $\ln n \geq 2(\ell-1)\ln \ln n$ 
                \item \label{assump:b} $\left(\ln n\right)^{-\frac{\ell-1}{n-1}} > \frac{1}{2}$
            \end{enumerate}
            In this case, we have
            % \begin{align*}
            %     t-\ell+1 &= n\left(\ln n + (\ell-1)\ln \ln n -x\right) - \ell +1 \\
            %     &\geq n\left(\ln n - (\ell-1)\ln \ln n \right) - (\ell -1) \\
            %     &\geq n\left(\ln n - 2(\ell-1)\ln \ln n \right) \\
            %     &\geq \frac{n}{2}\ln n,
            % \end{align*}
            % so,
            % \begin{align*}
            %     \sum_{r=0}^{\ell-1}\binom{t}{r}\left(\frac{1}{n-1}\right)^r 
            %     &\geq \binom{t}{\ell-1}\left(\frac{1}{n-1}\right)^{\ell-1} \\
            %     & \geq \frac{1}{(\ell-1)!}\left(\frac{t-\ell+1}{n-1}\right)^{\ell-1} \\
            %     &\geq \frac{1}{(\ell-1)!}\left(\frac{\frac{n}{2}\ln n}{n-1}\right)^{\ell-1} \\
            %     &\geq \frac{1}{(\ell-1)!2^{\ell-1}}\left(\ln n\right)^{\ell-1}.
            % \end{align*}
            \begin{align*}
                \sum_{r=0}^{\ell-1}\binom{t}{r}\left(\frac{1}{n-1}\right)^r 
                &\geq \binom{t}{\ell-1}\left(\frac{1}{n-1}\right)^{\ell-1} \\
                &\geq \frac{1}{(\ell-1)^{\ell-1}}\left(\frac{t}{n-1}\right)^{\ell-1} \\
                &= \frac{1}{(\ell-1)^{\ell-1}}\left(\frac{n(\ln n + (\ell-1)\ln\ln n  - x)}{n-1}\right)^{\ell-1} \\
                &\geq \frac{1}{(\ell-1)^{\ell-1}}\left(\ln n - (\ell-1)\ln\ln n\right)^{\ell-1} \\
                &\geq \frac{1}{(\ell-1)^{\ell-1}2^{\ell-1}}\left(\ln n\right)^{\ell-1} \;, \\
            \end{align*}
            where the first inequality just considers the last summand, the second inequality applies the bound $\binom{a}{b} \geq a^a/b^b$, and the last inequality is due to our assumption \ref{assump:a} on $n$.
            Thus,
            \begin{align*}
                ne^{-\frac{t}{n-1}}\sum_{r=0}^{\ell-1}\binom{t}{r}\left(\frac{1}{n-1}\right)^r 
                &\geq ne^{-\frac{t}{n-1}}\frac{1}{(\ell-1)^{\ell-1}2^{\ell-1}}\left(\ln n\right)^{\ell-1} \\
                &= n^{-\frac{1}{n-1}}e^{-\frac{n}{n-1}( (\ell-1)\ln\ln n-x )}\frac{1}{(\ell-1)^{\ell-1}2^{\ell-1}}\left(\ln n\right)^{\ell-1} \\
                &\geq \frac{1}{(\ell-1)^{\ell-1}2^{\ell}}e^x \left(\ln n\right)^{-\frac{\ell-1}{n-1}} \\
                &\geq \frac{1}{(\ell-1)^{\ell-1}2^{\ell+1}}e^x,
            \end{align*}
        
            where the last inequality holds due to assumption \ref{assump:b} on $n$.
        \end{enumerate}
        
        Therefore, for the cases where either $n>n_1$ or $n\leq n_1$ and $2(\ell-1)\ln \ln n\leq x$,
        \begin{align*}
            \mathbb{E}\left[W^{(t)}\right] 
            &\geq ne^{-\frac{t}{n-1}}\sum_{r=0}^{\ell-1}\binom{t}{r}\left(\frac{1}{n-1}\right)^r \\
            &\geq \min\left\{\frac{1}{2}e^{\frac{x}{2}}, \frac{1}{(\ell-1)^{\ell-1}2^{\ell+1}}e^x\right\} \\
            &\geq \frac{1}{(\ell-1)^{\ell-1}2^{\ell+1}}e^{\frac{x}{2}} \;,
        \end{align*}
        and therefore,
         \begin{align*}
             \Pr\left[Y_{\max}^{(n, \ell)}\leq- x\right] 
            &\leq \Pr\left[W^{(t)}=0\right] \\
            &\leq \frac{1}{\mathbb{E}\left[W^{(t)}\right]} \\
            &\leq (\ell-1)^{\ell-1}2^{\ell+1}e^{-\frac{x}{2}}.
         \end{align*}
        Lastly, for the case where $x
        \leq 2(\ell-1)\ln\ln n$ and $2\leq n\leq n_1$, it holds that
        $ x
        \leq 2(\ell-1)\ln\ln n 
        \leq 2(\ell-1)\ln\ln n_1$, and \begin{align*}
            \Pr\left[Y_{\max}^{(n, \ell)}\leq- x\right] 
        \leq 1 
        = \left(\ln n_1\right)^{(\ell-1)}e^{-(\ell-1)\ln\ln n_1} \leq \left(\ln n_1\right)^{(\ell-1)}e^{-\frac{x}{2}}.
        \end{align*}
        We conclude the proof of the lemma by noting that for $C_1(\ell)\triangleq\max\{(\ell-1)^{\ell-1}2^{\ell+1}, (\ln n_1)^{(\ell-1)}\}$ and any $x\geq 0, n\geq 2$, it holds that $\Pr\left[Y_{\max}^{(n, \ell)}<- x\right] \leq C_1(\ell)e^{-\frac{x}{2}}$ as required.
    \end{proof}

    \begin{corollary}\label{cor:multi-set-unif-int}
        The sequence $\left\{Y_{\max}^{(n,\ell)}\right\}_{n\geq2}$ is uniformly integrable.
    \end{corollary}
    \begin{proof}
        By \Cref{clm:l-sets-right-tail} and \Cref{clm:l-sets-left-tail}, we have that,
        \begin{align*}
            \Pr\left[ \left|Y_{\max}^{(n,\ell)}\right| > x \right] \leq \Pr\left[Y_{\max}^{(n, \ell)} >x\right] + \Pr\left[Y_{\max}^{(n, \ell)}\leq - x\right] \leq \left(C_0+C_1\right)e^{-\frac{x}{2}} \;.
        \end{align*}
        The claim follows according to \Cref{lem:exp-tail-impl-unif-conv}.
    \end{proof}

    Now we will prove \Cref{thm:multi-copy-no-redun} which is restated for convenience.
    \MultipleCopyNoRedund*
    
    \begin{proof}
        By \Cref{cor:multi-set-max-conv-in-dist} and \Cref{cor:multi-set-unif-int}, $\calY_{\max} = \left\{Y_{\max}^{(n, \ell)}\right\}_{n\geq 2}$ is uniformly integrable and converges in distribution to $\text{Gumbel}\left(\ln\left(\frac{m}{(\ell-1)!}\right), 1\right)$. Therefore, by \Cref{thm:conv-of-expec}, 
        \begin{align*}
            \lim_{n\to\infty}\mathbb{E}\left[Y_{\max}^{(n, \ell)}\right] 
            = \mathbb{E}\left[\text{Gumbel}\left(\ln\left(\frac{m}{(\ell-1)!}\right), 1\right)\right] 
            =\ln\left(\frac{m}{(\ell-1)!}\right) + \gamma.
        \end{align*}
        We note that $X_{\max}^{(n, \ell)} = n\ln n+ n(\ell-1)\ln\ln n + nY_{\max}^{(n, \ell)} $. Hence,
        \begin{align*}
            \mathbb{E}[X_{\max}^{(n, \ell)}] 
            &= n\ln n+ n(\ell-1)\ln\ln n + n\mathbb{E}\left[Y_{\max}^{(n, \ell)}\right] \\
            % &= n\ln\left(\frac{mn}{(\ell-1)!}\right) + n(\ell-1)\ln\ln n + \gamma n \pm o(n) \;.
            &= n \left(\ln\left(nm\right) + \left(\ell-1\right)\ln \ln n+ C_{\ell} \pm o(1) \right)
        \end{align*}
        where $C_{\ell} = \gamma - \ln((\ell-1)!)$. 
    \end{proof}

\section{Proof of \Cref{thm:mds-version}}
The proof of \Cref{thm:mds-version} is based on the same general framework as the proofs of \Cref{thm:single-copy-no-redun} and \Cref{thm:multi-copy-no-redun}, namely, analyzing normalized versions of the relevant random variables and establishing their convergence in distribution and uniform integrability. Having said that, unlike the previous theorems, there is no well known normalization related to the problem. Consequently we had to find it and prove its limiting distribution.

% Before we continue, lets recall our theorem.
% \MDSVersion*
Throughout this section, we adopt the following notation. Let $A^{(0)} = {\boldsymbol{0}}^{n\times (m+\rho)}, A^{(1)}, A^{(2)}, \ldots$ be a sequence of matrices constructed as follows: for each $t \in \mathbb{N}$, we draw $(v_1, \ldots, v_{m+\rho}) \sim \mathrm{Unif}([n]^{m+\rho})$ and set
    \[
        A^{(t)} = A^{(t-1)} + \left[\bfe_{v_1}, \ldots, \bfe_{v_{m+\rho}}\right] \;,
    \]
    where $\bfe_a$ denotes the column vector with a $1$ in the $a$-th position and $0$ in all other entries. Define $T_{n,m,\rho}$ as the random variable that corresponds to the smallest $t$ such that every row of $A^{(t)}$ contains at least $m$ nonzero entries.
    For each row $i\in[n]$ and time $t\ge 0$, let $X_i(t)$ denote the \emph{number of nonzero entries in row $i$ of $A^{(t)}$}, and by $E_i(t)$ the event that row $i$ has less than $m$ nonzero entries at time $t$, i.e., $E_i(t)=\{(X_1(t), \ldots, X_n(t))\in ([m+\rho]\cup \{0\})^n \mid X_i(t)<m\}$. In addition, denote $Z_{n,m,\rho} \triangleq\frac{\rho+1}{n}\,T_{n,m,\rho} - \ln n$. 

The proof of \Cref{thm:mds-version} follows a three-step argument. We begin by showing in \Cref{thm:mds-dist-conv} that $\{Z_{n,m,\rho}\}_{n\geq 2}$ converges in distribution to a Gumbel law with parameters $\mu = \ln \binom{m+\rho}{m-1}$ and $\beta=1$. We then establish in \Cref{clm:mds-right-tail} and in \Cref{clm:mds-left-tail} suitable tail bounds of $\{Z_{n,m,\rho}\}_{n\geq 2}$, which yield uniform integrability in \Cref{cor:mds-unif-int}. Combining these ingredients, we conclude the proof of \Cref{thm:mds-version} via \Cref{thm:conv-of-expec}.

We begin with some estimations which will be used in the proof of \Cref{thm:mds-dist-conv}.

\begin{lemma}\label{lem:intersection-asym}
    Let $m,\rho,k\in \mathbb{N}, x\in \mathbb{R}$ be fixed constants with respect to $n\in \Bbb{N}$, and set 
    \[
        \psi(x)=\binom{m+\rho}{m-1}e^{-x}.
    \]
    If $t=\frac{n}{\rho+1}(\ln n +x)$, then\footnote{Since $E_i(t) = \{(X_1(t), \ldots, X_n(t))\in ([m+\rho]\cup \{0\})^n \mid X_i(t) < m\}$, the event $E_1(t)\cap\cdots\cap E_k(t)$ is the set $\{(X_1(t), \ldots, X_n(t))\in ([m+\rho]\cup \{0\})^n \mid \forall i\in [k]: X_i(t) < m\}$.
    }
    \begin{equation}\label{eq:intersection-asym}
    \Pr\left[E_1(t)\cap\cdots\cap E_k(t) \right]
    = \frac{\psi(x)^k}{n^k} (1\pm o(1)).
    \end{equation}
    Consequently, if
    \[
    U_k(n)\,\triangleq\,\sum_{1\le i_1<\cdots<i_k\le n}
    \Pr\!\Big(\bigcap_{j=1}^k E_{i_j}(t)\Big),
    \]
    then
    \begin{equation}\label{eq:Sk-asym}
    U_k(n)= \frac{\psi(x)^k}{k!} (1\pm o(1)) \,.
    \end{equation}
\end{lemma}
\begin{proof}
    Fix $k\ge 1$ and consider the event $E_1(t)\cap\cdots\cap E_k(t)$ which is the event where at time $t$, every row $i\in [k]$, contains less than $m$ nonzero entries.
    In order to calculate $\Pr\!\Big(E_1(t)\cap\cdots\cap E_k(t)\Big)$, we will estimate the probabilities of different possible patterns of $A^{(t)}$. For each $j\in [m+\rho]$, denote by $\xi_j\subseteq [k]$ the set of \textbf{rows} with zero entries at column $j$, and call $\boldsymbol{\xi}=(\xi_1,\dots, \xi_{m+\rho})$ the configuration at time $t$. A configuration $\boldsymbol{\xi}$ for which the event $E_1(t)\cap\cdots\cap E_k(t)$ occurs, it called a \emph{valid configuration}.
    
    For $t=\tfrac{n}{\rho+1}(\ln n + x)$, and each fixed $\xi\subseteq[k], j\in [m+\rho]$,
    \begin{align*}
        \Pr\left[\text{At time $t$, $\xi_j = \xi$}\right]
        =& \Pr\left[\forall i\in \xi: A^{(t)}_{i,j}=0\right]\Pr\left[\forall i\in [k]\setminus\xi: A^{(t)}_{i,j}\geq 1 \mid \forall i\in \xi: A^{(t)}_{i,j}=0\right]\\
        =& \Big(1- \frac{|\xi|}{n}\Big)^{t} \sum_{a=0}^{k-|\xi|}(-1)^a\binom{k-|\xi|}{a}\left(1-\frac{a}{n - |\xi|}\right)^t \\
        =& \sum_{a=0}^{k-|\xi|}(-1)^a\binom{k-|\xi|}{a}\left(1-\frac{|\xi| + a}{n}\right)^t\\
        =& \sum_{a=0}^{k-|\xi|}(-1)^a\binom{k-|\xi|}{a}\left(1-\frac{|\xi| + a}{n}\right)^{\frac{n}{\rho+1}(\ln n +x)}\\
        =& (n^{-1}e^{-x})^{|\xi|/(\rho+1)}(1\pm o(1))\;.
    \end{align*}
    Here, the last equality holds since $\left(1-\frac{|\xi| + a}{n}\right)^{\frac{n}{\rho+1}(\ln n +x)} = (n^{-1}e^{-x})^{(|\xi|+a)/(\rho+1)}(1\pm o(1))$, the dominant term in the sum is obtained when $a=0$, namely, $(n^{-1}e^{-x})^{|\xi|/(\rho+1)}(1\pm o(1))$, while all other terms (in absolute value) are $o((n^{-1}e^{-x})^{|\xi|/(\rho+1)})$.

    As the columns are independent, the probability that at time $t$, the obtained configuration is $\boldsymbol{\xi}=(\xi_1,\dots, \xi_{m+\rho})$ is 
    \begin{equation}\label{eq:monomial-term}
        \prod_{j=1}^{m+\rho} \big(n^{-1}e^{-x}\big)^{|\xi_j|/(\rho+1)}(1\pm o(1))
        \;=\;
        \big(n^{-1}e^{-x}\big)^{(\sum_{j=1}^{m+\rho} |\xi_j|)/(\rho+1)} (1\pm o(1))\;.
    \end{equation}
    
    Now, note that for the event $E_1(t)\cap\cdots\cap E_k(t)$ to occur, for any $i\in [k]$, we must have $\sum_{j=1}^{m+\rho}\mathds{1}_{i\in \xi_j}\,\ge\, \rho+1$. 
    Therefore, it holds that $\sum_{i=1}^{k}\sum_{j=1}^{m+\rho}\mathds{1}_{i\in \xi_j}\ge\sum_{i=1}^{k}(\rho+1)=k(\rho+1)$. However, by changing the order of summation, we get 
    \begin{align*}
        k(\rho+1) \leq \sum_{i=1}^{k}\sum_{j=1}^{m+\rho}\mathds{1}_{i\in \xi_j}
     = \sum_{j=1}^{m+\rho}\sum_{i=1}^{k}\mathds{1}_{i\in \xi_j} = \sum_{j=1}^{m+\rho}|\xi_j|.
    \end{align*} 
    If $\sum_{j=1}^{m+\rho}|\xi_j| = k (\rho + 1)$, then~\eqref{eq:monomial-term} becomes $n^{-k}e^{-kx} (1 \pm o(1))$. Also observe that $\boldsymbol{\xi}$ is valid and  $\sum_{j=1}^{m+\rho}|\xi_j| = k (\rho + 1)$ if and only if for all $i\in [k]$, it holds that $\sum_{j=1}^{m+\rho}\mathds{1}_{i\in \xi_j} = \rho+1$. Thus, there are exacly $\binom{m+\rho}{\rho +1}^k = \binom{m+\rho}{m-1}^k$ valid configurations $\boldsymbol{\xi}$ for which $\sum_{j=1}^{m+\rho}|\xi_j| = k (\rho + 1)$, and we have get that at time $t$,
    \begin{align*}
        \Pr\left[ \boldsymbol{\xi} \text{ is valid and }\sum_{j=1}^{m+\rho}|\xi_j| = k (\rho + 1) \right] = \binom{m+\rho}{m-1}^{\,k}\, n^{-k} e^{-kx}(1\pm o(1))
        = \frac{\psi(x)^k}{n^k}(1\pm o(1))\;.
    \end{align*}
    where $\psi(x) = \binom{m + \rho}{m-1} e^{-x}$.
    
    If $\sum_{j=1}^{m+\rho}|\xi_j| > k (\rho + 1)$, then \eqref{eq:monomial-term} is $o(n^{-k})$. As $k, m$, and $\rho$ are constants and the number of total configurations is $2^{k(m + \rho)}$, we get that 
    \[
    \Pr\left[ \boldsymbol{\xi} \text{ is valid and }\sum_{j=1}^{m+\rho}|\xi_j| > k (\rho + 1) \right] \leq 2^{k(m+\rho)}\cdot o(n^{-k}) = o(n^{-k})
    \;.
    \]

    Thus, since we cannot have a valid configuration $\boldsymbol{\xi}$ with $\sum_{j=1}^{m+\rho}|\xi_j| < k (\rho + 1)$, we conclude that
    \[
        \Pr\left[E_1(t)\cap\cdots\cap E_k(t)\right]
        = \frac{\psi(x)^k}{n^k}(1\pm o(1))\,,
    \]
    as claimed in Lemma~\ref{lem:intersection-asym}.
    Lastly, since $\Pr\left[\bigcap_{j=1}^k E_{i_j}(t)\right] = \Pr\left[\bigcap_{j=1}^k E_{j}(t)\right]$ is the same for any set $\{i_1, \ldots,i_k\} \subset [n]$, 
    if $U_k(n)\,\triangleq\,\sum_{1\le i_1<\cdots<i_k\le n}
    \Pr\!\Big(\bigcap_{j=1}^k E_{i_j}(t)\Big)$, then $U_k(n) = \binom{n}{k}\frac{\psi(x)^k}{n^k}(1\pm o(1)) = \frac{\psi(x)^k}{k!}(1\pm o(1))$ as stated. 
\end{proof}

The proof of \Cref{thm:mds-dist-conv} uses the asymptotic evaluation we analyzed in \Cref{lem:intersection-asym} and Bonferroni’s inequalities to prove the convergence in distribution of $\{Z_{n,m,\rho}\}_{n\geq 2}$.
\begin{theorem}\label{thm:mds-dist-conv}
    For each fixed $x\in\mathbb{R}$,
    \begin{align*}
      \lim_{n\to\infty}
      \Pr\!\left[T_{n,m,\rho}\le \frac{n}{\rho+1}\left(\ln n + x\right)\right]
      \;=\;
      \exp\!\left(-\binom{m+\rho}{m-1}e^{-x}\right).
    \end{align*}
    Equivalently, for $Z_{n,m,\rho} \triangleq \frac{\rho+1}{n}\,T_{n,m,\rho} - \ln n$, we have 
    \begin{align*}
        Z_{n,m,\rho} \;\xrightarrow{\cal D}\;
        {\rm Gumbel}\!\left(\,\ln\!\binom{m+\rho}{m-1},\,1\,\right).
    \end{align*}
\end{theorem}
\begin{proof}
    Recall that
    \[
    \{T_{n,m,\rho}\le t\}
    =\Big(\bigcup_{i=1}^n E_i(t)\Big)^{\!c},
    \]
    since $T_{n,m,\rho}\le t$ means that all rows of $A^{(t)}$ have at least $m$ nonzero entries. 
    Therefore,
    \[
    \Pr\left[T_{n,m,\rho}\le t\right]
    = 1 - \Pr\!\left[\bigcup_{i=1}^n E_i(t)\right].
    \] 
    By Bonferroni’s inequalities, for every $\frac{n}{2}\ge L\ge 1$,
    \[
    \sum_{k=1}^{2L}(-1)^{k+1} U_k(n)
    \le
    \Pr\!\left[\bigcup_{i=1}^n E_i(t)\right]
    \le
    \sum_{k=1}^{2L-1}(-1)^{k+1} U_k(n) \;,
    \]
    where recall that $U_k(n)$ is $\sum_{1\le i_1<\cdots<i_k\le n}
    \Pr\!\left[\bigcap_{j=1}^k E_{i_j}(t)\right]$. 
    For any fixed $L$, there exists $n_0(L)$ such that for all $n \geq n_0(L)$, we have
    % Fix $L$.  
    % By the asymptotic formula \eqref{eq:Sk-asym}, for any  sufficiently large $n$,
    \[
    \Bigg|
    U_k(n)-\frac{\psi(x)^k}{k!}
    \Bigg|
    \le \frac{1}{2L^2}
    \qquad\text{for all }1\le k\le 2L.
    \]
    Hence,
    \[
    \sum_{k=1}^{2L}(-1)^{k+1}\frac{\psi(x)^k}{k!}-\frac{1}{L}
    \;\le\;
    \Pr\!\left[\bigcup_{i=1}^n E_i(t)\right]
    \;\le\;
    \sum_{k=1}^{2L-1}(-1)^{k+1}\frac{\psi(x)^k}{k!}+\frac{1}{L}\;.
    \]
    As $L\to\infty$, both bounds converge to the alternating series 
    \[
    \sum_{k=1}^\infty (-1)^{k+1}\frac{\psi(x)^k}{k!}
    = 1 - e^{-\psi(x)}.
    \]
    Therefore, %\rc{We should try to be more precise here with the asymptotics. I am not sure that it is precise to write here the $o(1)$ since it is w.r.t $n$ and $L$ is fixed with respect to $n$. 
    % Maybe we can just skip this line and go directly to what finishes the proof.
    % What do you think?}\al{hmm yeah I see what you're saying. I think we can just skip it. The argument I have in mind is as follows: 
    % Let $\epsilon>0$. 
    % There exist $L_0(\epsilon)$ s.t. 
    % \[\left| \sum_{k=1}^{2L_0}(-1)^{k+1}\frac{\psi(x)^k}{k!}-\frac{1}{L_0} - \left(1 - e^{-\psi(x)}\right)\right| \leq \epsilon  \qquad 
    % \text{and} 
    % \qquad \left| \sum_{k=1}^{2L_0-1}(-1)^{k+1}\frac{\psi(x)^k}{k!}+\frac{1}{L_0} - \left(1 - e^{-\psi(x)}\right)\right| \leq \epsilon ,\]
    % (and these are independent of $n$ so we're good) hence for any $n\geq n_0(L_0)$ (so it actually only depends on $\epsilon$) it holds that 
    % \[1 - e^{-\psi(x)} - \epsilon \leq  \Pr\!\left[\bigcup_{i=1}^n E_i(t)\right] \leq 1 - e^{-\psi(x)} + \epsilon, \] 
    % i.e. $\left| \Pr\!\left[\bigcup_{i=1}^n E_i(t)\right] - \left(1 - e^{-\psi(x)}\right)\right| \leq \epsilon $ or $\lim_{n\to\infty}\Pr\!\left[\bigcup_{i=1}^n E_i(t)\right]
    % =
    % 1 - e^{-\psi(x)}$. Does this make sense? This is also why $k$ is fixed in the previous lemma - as fixed $\epsilon$ fixes $L_0$ which fixes all values of required $k$s - hope that is clear.} \rc{I agree. To keep things simpler, we can remove the next line.}
    % \[
    % \Pr\!\left[\bigcup_{i=1}^n E_i(t)\right]
    % =
    % (1 - e^{-\psi(x)})(1\pm o(1)),
    % \]
    % and hence
    \[
    \lim_{n\to \infty}\Pr[T_{n,m,\rho}\le t]
    = 1 - \lim_{n\to \infty}\Pr\!\left[\bigcup_{i=1}^n E_i(t)\right]
    = e^{-\psi(x)},
    \]
    which completes the proof.
\end{proof}

We now head to the proof that the set $\{Z_{n,m,\rho}\}_{n\geq 2}$ is uniformly integrable.

\begin{claim}\label{clm:mds-right-tail} For any $x\geq 0$ and any $n\in \mathbb{N}$:
    \begin{align*}
        \Pr[Z_{n,m,\rho} > x] \leq 2^{m+\rho}e^{-x}
    \end{align*}
\end{claim}
\begin{proof}
    Fix $x\in\mathbb{R}$ and set $u=\ln n+x$, so that $\Pr[Z_{n,m,\rho} > x] = \Pr[T_{n,m,\rho} > \frac{n}{\rho+1}u]$. 
    Recall that we define the random variable $X_i(t)$ to represent the number of columns that were collected in the $i$-th row by time $t$. We have that $X_i(t)\sim \textup{Bin}(m+\rho,p_t)$ where $p_t=1-\left(1-\frac{1}{n}\right)^t$ is the probability that a specific column draws a specific row by time $t$.
    %Recall that for any $t\ge 0, i\in [n]$, $X_i(t)\sim \textup{Bin}(m+\rho,p_t)$ where $p_t=1-\left(1-\frac{1}{n}\right)^t$ is the probability that a column draws a specific row by time $t$. Hence, 
    \begin{align*}
        \Pr[X_i(t)<m] 
        = \sum_{k=0}^{m-1}\binom{m+\rho}{k}\,p_t^{\,k}(1-p_t)^{m+\rho-k}
        \le (1-p_t)^{\rho+1}\sum_{k=0}^{m-1}\binom{m+\rho}{k} 
        \le 2^{m+\rho}(1-p_t)^{\rho+1}.
    \end{align*}
    With $t=\tfrac{n}{\rho+1}\,u$ this gives
    \begin{align*}
        1-p_t
        = \left(1-\frac{1}{n}\right)^{t}
    \;\le\;
    e^{-t/n}
    = e^{-u/(\rho+1)}.
    \end{align*}
    Therefore,
    \begin{align*}
        \Pr\left[E_i(t)\right] = 
        \Pr\left[X_i(t)<m\right]
        \le 2^{m+\rho}(1-p_t)^{\rho+1}
        \le 2^{m+\rho} e^{-u}
        = 2^{m+\rho}n^{-1}e^{-x}.
    \end{align*}
    Applying the union bound over all $n$ rows yields
    \begin{align*}
        \Pr[Z_{n,m,\rho} > x] = \Pr\left[T_{n,m,\rho} > \frac{n}{\rho+1}u\right] \leq 2^{m+\rho}e^{-x}\;.
    \end{align*}
\end{proof}

For the proof of the left tail, we recall some properties of NA random variables,
\begin{proposition}[Property P$_6$, \cite{joag1983negative}]\label{prop:na_p6}
    Non-decreasing functions defined on disjoint subsets of a set of NA random variables are NA.
\end{proposition}
\begin{proposition}[Property P$_7$, \cite{joag1983negative}]\label{prop:na_p7}
    The union of independent sets of NA random variables is NA.
\end{proposition}

\begin{claim}\label{clm:mds-left-tail} For any $x\geq 0$ and $2\leq n\in \mathbb{N}$:
    \begin{align*}
        \Pr[Z_{n,m,\rho} \leq -x] \leq \frac{1}{C_2(m,\rho)}e^{-x}
    \end{align*}
    where $C_2(m,\rho) = \frac{1}{2}\binom{m+\rho}{m-1}(1-e^{-\frac{m}{m+\rho}})^{m-1}$. 
\end{claim}
\begin{proof}
    Fix $x,n$, and set $t=\frac{n}{\rho+1}(\ln n -x)$. If $x > \ln n - \frac{(\rho+1)m}{m+\rho}$ then $ t=\frac{n}{\rho+1}(\ln n -x) < \frac{nm}{m+\rho}$, which is the minimal number of rounds in order to have at least $m$ nonzero entries in each row, and in particular $\Pr[Z_{n,m,\rho} \leq -x] = \Pr[T_{n,m,\rho} \leq t] = 0 \leq Ce^{-x}$ for any $C>0$.
    Assume that $0\leq x\leq \ln n - \frac{(\rho+1)m}{m+\rho}$ and let $W^{(t)}$ be the number of rows with less than $m$ nonzero entries in $A^{(t)}$: $W^{(t)}=\sum_{i=1}^n I_i^{(t)}$ where $I_i^{(t)}$ is the indicator that row $i$ of $A^{(t)}$ has less than $m$ nonzero entries. We have that 
    \begin{align*}
         \Pr[Z_{n,m,\rho} \leq -x] 
         = \Pr[T_{n,m,\rho} \leq t] 
         = \Pr[W^{(t)}=0].
     \end{align*}

     For every $i_1\ne i_2$, $\Cov(I_{i_1}^{(t)},I_{i_2}^{(t)}) \leq 0 $:
     Similarly to the proof of \Cref{clm:l-sets-left-tail}, for any fixed $j$, by \Cref{thm:neg-assoc} the set of entries $\{A_{i,j}^{(t)}\}_{i\in [n]}$ is NA. For any $j_1\ne j_2$ the sets are independent, hence, by \Cref{prop:na_p7} the entire set $\{A_{i,j}^{(t)}\}_{i\in [n], j\in [m+\rho]}$ is NA. For any $i\in [n], j\in [m+\rho]$ let $I_{i,j}^{(t)}$ be the indicator that $A_{i,j}^{(t)}$ is nonzero. Since these are non-decreasing functions defined on disjoint subsets of a set of NA random variables, the set $\{I_{i,j}\}_{i\in [n], j\in [m+\rho]}$ is NA. Note that $X_i(t) = \sum_{j=1}^{m+\rho} I_{i,j}^{(t)}$, therefore by \Cref{prop:na_p6} $\left\{X_i(t)\right\}_{i\in [n]}$ are NA. 
     Finally, since the indicator function $f(X)=\mathds{1}_{X<m}$ is non-increasing, and $I_i^{(t)} = \mathds{1}_{X_i(t)<m}$, by the definition of NA (\Cref{def:neg-assoc}), we have $\Cov(I_{i_1}^{(t)},I_{i_2}^{(t)}) \leq 0$ and thus $0< \Var(W^{(t)}) \leq \mathbb{E}[W^{(t)}] $. Moreover, 
    \begin{align*}
        \mathbb{E}[W^{(t)}] = \sum_{j=1}^n \Bbb{E}\left[I_j^{(t)}\right]&= n\cdot \Pr\left[ X_i(t) < m\right]\\ &= n\cdot\sum_{k=0}^{m-1}\binom{m+\rho}{k}p_t^k\left(1-p_t\right)^{m+\rho-k} \\
        &\geq n\binom{m+\rho}{m-1}p_t^{m-1}\left(1-p_t\right)^{m+\rho-(m-1)} \\
        &= n\binom{m+\rho}{m-1}p_t^{m-1}\left(1-p_t\right)^{\rho+1} \;.
    \end{align*}
    In addition, 
    \begin{align*}
        \left(1-p_t\right)^{\rho+1} = \left(1-\frac{1}{n}\right)^{(\rho+1)t} \geq e^{-\frac{(\rho+1)t}{n-1}} = e^{-\frac{n}{n-1}(\ln n - x)} \geq \frac{1}{2n}e^x,
    \end{align*}
    and since we assume that $0\leq x\leq \ln n - \frac{(\rho+1)m}{m+\rho}$,
    \begin{align*}
        p_t = 1- \left(1-\frac{1}{n}\right)^t \geq 1-e^{-\frac{t}{n}} 
        = 1-e^{-\frac{1}{\rho+1}(\ln n -x)} 
        \geq 1-e^{-\frac{m}{m+\rho}} > 0.
    \end{align*}
    Thus, 
    \begin{align*}
        \mathbb{E}[W^{(t)}] 
        \geq n\binom{m+\rho}{m-1}p_t^{m-1}(1-p_t)^{\rho+1}
        \geq \frac{1}{2}\binom{m+\rho}{m-1}(1-e^{-\frac{m}{m+\rho}})^{m-1}e^x.
    \end{align*}
    
    Therefore, for $C_2(m,\rho)\triangleq \frac{1}{2}\binom{m+\rho}{m-1}(1-e^{-\frac{m}{m+\rho}})^{m-1}$ we obtain that
     \begin{align*}
         \Pr[Z_{n,m,\rho} \leq -x] 
         &= \Pr[W^{(t)}=0] \\
         &\leq \Pr\left[|W^{(t)}- \mathbb{E}[W^{(t)}]|
         \geq \mathbb{E}[W^{(t)}]\right] \\
         &\leq \frac{\Var(W^{(t)})}{\mathbb{E}[W^{(t)}]^2} \\
         &\leq \frac{1}{\mathbb{E}[W^{(t)}]} \\
         &\leq \frac{1}{C_2(m,\rho)}e^{-x}.
     \end{align*}
     as claimed.
\end{proof}

\begin{corollary}\label{cor:mds-unif-int}
        The sequence $\left\{Z_{n,m,\rho}\right\}_{n\geq 2}$ is uniformly integrable.
    \end{corollary}
\begin{proof}
        By \Cref{clm:mds-right-tail} and \Cref{clm:mds-left-tail}, for every $x\in \mathbb{R}, n\geq 2$ we have that,
        \begin{align*}
            \Pr\left[ \left|Z_{n,m,\rho}\right| > x \right] \leq \Pr\left[Z_{n,m,\rho} >x\right] + \Pr\left[Z_{n,m,\rho}\leq- x\right] \leq \left(2^{m+\rho} + \frac{1}{C_2(m,\rho)}\right)e^{-x} \;.
        \end{align*}
        The corollary follows according to \Cref{lem:exp-tail-impl-unif-conv}.
    \end{proof}

We can now complete the proof of \Cref{thm:mds-version}, restated below for clarity.
\MDSVersion*
\begin{proof}
    According to \Cref{cor:mds-unif-int}, and \Cref{thm:mds-dist-conv}, the sequence $\left\{Z_{n,m,\rho}\right\}_{n\geq 2}$ is uniformly integrable and converges in distribution to $\Gumbel(\ln \binom{m+\rho}{m-1},1)$. Therefore, by \Cref{thm:conv-of-expec}, we have that
    \begin{align*}
        \lim_{n\to\infty}\mathbb{E}\left[Z_{n,m,\rho}\right] = \ln \binom{m+\rho}{m-1} + \gamma.
    \end{align*}
    Recall that $T_{n,m,\rho} = \frac{n}{\rho+1}\left(Z_{n,m,\rho}+\ln n\right)$, hence, 
    \begin{align*}
        \mathbb{E}\left[T_{n,m,\rho}\right] = \frac{n}{\rho+1} \ln n + \frac{n}{\rho+1}\ln \binom{m+\rho}{m-1} + \frac{n}{\rho+1}\gamma \pm o(n).
    \end{align*}
\end{proof}

\section{Proof of \Cref{thm:regenerating-ex-version}}
We start by recalling the statement of the theorem.
%\defBadConf*

\RegeneratingVersion*
We will first prove the following upper bound
\begin{claim}
    For every positive $x$, we have 
    \[
    \Pr \left[T \geq \frac{n}{\alpha^*}\cdot(\ln n +x)\right] \leq   \frac{\beta^*}{b} \cdot e^{-x} \left( 1 + (mb)^2 n^{-\frac{1}{\alpha^*}} \right) \;.
    \]
\end{claim}
\begin{proof}
    Define $t \triangleq \frac{n}{\alpha^*}(\ln n + x)$. We begin by focusing only on the first block, i.e., the case $a=1$. Our goal is to upper bound the probability that, at time $t$, the collection of all nonzero entries is still a bad set, that is, lies in $\calB$. This will ensure that the block is not yet in a recovering state. More precisely, we will prove that 
    \[
    \Pr\left[ \left \lbrace (i,j) \in \left[b\right] \times [m] \mid A_{i,j}^{(t)} \neq 0 \right \rbrace \in \calB \right] \leq \beta^* \cdot n^{-1}e^{-x} \left( 1 + (mb)^2 n^{-\frac{1}{\alpha^*}} \right)\;.
    \]
    Then, the claim will follow from taking a union bound over all the $n/b$ blocks.
    
    It holds that 
    \[
    \Pr\left[ \left \lbrace (i,j) \in \left[b \right] \times [m] \mid A_{i,j}^{(t)} \neq 0 \right \rbrace \in \calB \right] = \sum_{B\in \calB} \Pr\left[ \left \lbrace (i,j) \in \left[b\right] \times [m] \mid A_{i,j}^{(t)} \neq 0 \right \rbrace  = B\right] \;.
    \]
    Fix $B\in \calB$ and denote by $\overline{B}$ its complement in the block, i.e., $B \cup \overline{B} = [b] \times [m]$ and $B\cap \overline{B} = \emptyset$. We have 
    \[
    \Pr\left[ \left \lbrace (i,j) \in \left[b \right] \times [m] \mid A_{i,j}^{(t)} \neq 0 \right \rbrace  = B\right] = \Pr\left[ \left \lbrace (i,j) \in \left[b \right] \times [m] \mid A_{i,j}^{(t)} = 0 \right \rbrace  = \overline{B}\right] \;,
    \]
    and if we denote $s_j \triangleq |\{ i\in [b] \mid (i,j) \in \overline{B} \}|$, then,
    \[
        \Pr\left[ \left \lbrace (i,j) \in \left[b \right] \times [m] \mid A_{i,j}^{(t)} = 0 \right \rbrace = \overline{B}\right] = \prod_{j=1}^m \left( 1 - \frac{s_j}{n}\right)^t \leq \prod_{j=1}^m \left( 1 - \frac{1}{n}\right)^{s_j t} = \left( 1 - \frac{1}{n} \right)^{(mb - |B|) \cdot t} \;.
    \]
    Now, observe that this upper bound depends only on the size of $B$. Thus,
    \[
    \Pr\left[ \left \lbrace (i,j) \in \left[b \right] \times [m] \mid A_{i,j}^{(t)} \neq 0 \right \rbrace \in \calB\right] 
    \leq \sum_{B\in \calB} \left( 1 - \frac{1}{n} \right)^{(mb - |B|) \cdot t} 
    \leq \sum_{B\in \calB} e^{-\frac{(mb - |B|)(\ln n + x)}{\alpha^*}} \;.
    \]
    Note that $mb - |B|\geq \alpha^*$. For all $s\in[0,mb]$, denote by $\beta_s \triangleq | \{B\in \mathcal{B} \mid |B| = s\}|$, i.e., the number of the sets in $\calB$ of size exactly $s$. Furthermore, set $s^* \triangleq mb-\alpha^*$. Then, 
    \begin{align*}
        \sum_{B\in \calB} e^{-\frac{(mb - |B|)(\ln n + x)}{\alpha^*}}&= \sum_{s=s^*}^{mb} \beta_s \cdot e^{-\frac{(mb - s)(\ln n + x)}{\alpha^*}} \\
        &\leq \beta_{s^*} \cdot n^{-1}e^{-x} + \sum_{s=s^* + 1}^{mb} \beta_s \cdot e^{-\frac{\alpha^* +1}{\alpha^*} \cdot(\ln n + x)}\\
        &\leq \beta_{s^*} \cdot n^{-1}e^{-x} \left( 1 + (mb)^2 n^{-\frac{1}{\alpha^*}} \right)\;.
    \end{align*}
    The claim follows by a union bound over all $a\in [n/b]$.
\end{proof}

With this bound, we can prove \Cref{thm:regenerating-ex-version}.
\begin{proof}[Proof of \Cref{thm:regenerating-ex-version}]
    Define $t_0 = (n \ln n)/\alpha^*$. It holds that 
    \[
    \Bbb{E}[T] = \sum_{t\geq 1} \Pr[T \geq t] = \sum_{t=1}^{t_0-1} \Pr[T \geq t] + \sum_{t=t_0}^{\infty} \Pr[T \geq t] \;.
    \]
    The first sum can be bounded trivially by $t_0$ (every term in the sum is at most $1$). We focus on the second sum. 
    \begin{align*}
        \sum_{t=t_0}^{\infty} \Pr[T \geq t] &= \sum_{k=0}^{\infty} \Pr\left[ T \geq \frac{n \ln n}{\alpha^*} + k\right] \\
        &\leq \sum_{k=0}^{\infty}\Pr\left[ T \geq \frac{n}{\alpha^*} \left(\ln n + k \cdot \frac{\alpha^*}{n} \right)\right] \\
        &\leq \frac{\beta_{s^*}}{b} \cdot \left( 1 + (mb)^2 n^{-\frac{1}{\alpha^*}} \right) \cdot \sum_{k=0}^{\infty} e^{-k \cdot\frac{\alpha^*}{n}} \\
        &\leq \frac{\beta_{s^*}}{b} \cdot \left( 1 + (mb)^2 n^{-\frac{1}{\alpha^*}} \right) \cdot \frac{1}{1 - e^{-\frac{\alpha^*}{n}}} \;.
    \end{align*}
    Now, by Taylor expansion, $1-e^{-x} = x - x^{2}/2 +O(x^3)$,
    we conclude that
    \[
     \Bbb{E}[T] \leq \frac{n}{\alpha^*}\ln n + \frac{\beta_{s^*}}{b\alpha^*} \cdot n \cdot(1 + m^2bn^{-\frac{1}{\alpha^*}}) (1 + o(1)) \;,
    \]
    and the result follows by recalling that $m$ and $b$ are constant with respect to $n$.
\end{proof}

\bibliographystyle{alpha}
\bibliography{bibliography}
\end{document}